\documentclass[lettersize,journal]{IEEEtran}
\usepackage{graphicx}
\usepackage{amssymb}
\usepackage{setspace} 

\usepackage[T1]{fontenc}

\usepackage{booktabs}
\usepackage{booktabs,makecell,multirow}

\usepackage{booktabs}
\usepackage{hyperref}

\usepackage{bm}
\usepackage{color, xcolor}
\usepackage{transparent}
\usepackage{import}
\usepackage[ruled]{algorithm2e}
\usepackage{amsmath}
\usepackage{amsthm}
\allowdisplaybreaks[4]
\usepackage{algorithmic}
\makeatletter

\newcommand{\Rmnum}[1]{\expandafter\@slowromancap\romannumeral #1@}
\makeatother
\usepackage{subfigure} 
\usepackage{threeparttable}

\usepackage{color}
\usepackage{transparent}
\usepackage{graphicx}
\usepackage{import}

\newtheorem{theorem}{Theorem}

\newtheorem{corollary}{\bf Corollary}
\newtheorem{lemma}{Lemma}
\newtheorem{assumption}{\bf Assumption}
\ifCLASSINFOpdf
\else
\fi
 \usepackage{graphicx}
\usepackage{CJKutf8}
\usepackage{bm}
\usepackage{cite}
\usepackage{amssymb}
\usepackage{amsmath}
 \usepackage{threeparttable}
\begin{document}
%
\title{Optimizing Split Federated Learning with Unstable Client Participation}
\markboth{IEEE Transactions on  Mobile Computing}%
{Shell \MakeLowercase{\textit{et al.}}: Bare Demo of IEEEtran.cls for IEEE Journals}
%



	
\author{Wei Wei,~\IEEEmembership{Graduate Student Member,~IEEE,}
	Zheng~Lin,
 Xihui Liu,~\IEEEmembership{Member,~IEEE,}
 Hongyang~Du,~\IEEEmembership{Member,~IEEE,}
 Dusit~Niyato,~\IEEEmembership{Fellow,~IEEE,} and~Xianhao~Chen,~\IEEEmembership{Member,~IEEE}

\thanks{The work was supported in part by the Research Grants Council of Hong Kong under Grant 27213824 and CRS HKU702/24, in part by HKU-SCF FinTech Academy R\&D Funding, and in part by HKU IDS Research Seed Fund under Grant IDS-RSF2023-0012.  
\textit{(Corresponding author: Xianhao Chen)}
}
\thanks{ Wei Wei, Zheng Lin, Xihui Liu, Hongyang Du, and Xianhao Chen are with the Department of Electrical and Electronic Engineering, the University of Hong Kong, Pok Fu Lam, Hong Kong SAR, 999077, China. Xihui Liu and Xianhao Chen are also with the HKU Musketeers Foundation Institute of Data Science, the University of Hong Kong, Pok Fu Lam, Hong Kong SAR, 999077, China. Dusit Niyato is with the College of Computing
and Data Science, Nanyang Technological University, Singapore 639798. (e-mail: weiwei@eee.hku.hk; linzheng@eee.hku.hk; xihuiliu@eee.hku.hk; duhy@eee.hku.hk;
xchen@eee.hku.hk; dniyato@ntu.edu.sg).
}
}

\maketitle

\begin{abstract}
To enable training of large artificial intelligence (AI) models at the network edge, split federated learning (SFL) has emerged as a promising approach by distributing computation between edge devices and a server. However, while unstable network environments pose significant challenges to SFL, prior schemes often overlook such an effect by assuming perfect client participation, rendering them impractical for real-world scenarios. In this work, we develop an optimization framework for SFL with unstable client participation. We theoretically derive the first convergence upper bound for SFL with unstable client participation by considering activation uploading failures, gradient downloading failures, and model aggregation failures. Based on the theoretical results, we formulate a joint optimization problem for client sampling and model splitting to minimize the upper bound. We then  develop an efficient solution approach to solve the problem optimally. 
Extensive simulations on EMNIST and CIFAR-10 demonstrate the superiority of our proposed framework  compared to existing benchmarks.
\end{abstract}


%
\IEEEpeerreviewmaketitle
\section{Introduction}
Edge learning enables artificial intelligence (AI) computation close to the data source, enhancing privacy, saving bandwidth, and enabling personalization of edge devices~\cite{Edge-AI-Technology-Report, 10843329, 10040976,wei2025pipeliningsplitlearningmultihop, ZW_spectrum,9944188,9839238,10298247}. However, the rapid expansion of AI models presents unprecedented computational challenges. State-of-the-art models are inherently incompatible with the limited hardware resources of edge devices or servers. For instance, edge devices such as the NVIDIA Jetson AGX Xavier with 16 GB of RAM struggle to train models that exceed 50-million parameters with full precision \cite{jetson-agx}. Given these challenges, split federated learning (SFL)~\cite{vepakomma2018split, lin2023split,10934144}, which utilizes a client-edge server architecture, presents a promising solution. In SFL, the AI model is divided into two segments: the initial layers are trained on clients over local data, whereas the subsequent layers are trained on a server that provides more powerful computing resources. Moreover, a federated server (or simply Fed server) periodically aggregates client-side sub-models, akin to federated learning (FL) paradigm~\cite{9084352,lin2024fedsn,konevcny2016federated, 10233897,lin2024efficient,lin2025leo}, to synchronize models across multiple clients in parallel. SFL alleviates the hardware constraints of clients by offloading the most demanding computations to more powerful servers, thereby recognized as a promising framework to enable large-scale model training on edge devices.

\begin{figure}[t!]

         \subfigure[\centering Illustration of network failures in SFL.] {
\centering\includegraphics[width=4.1cm]{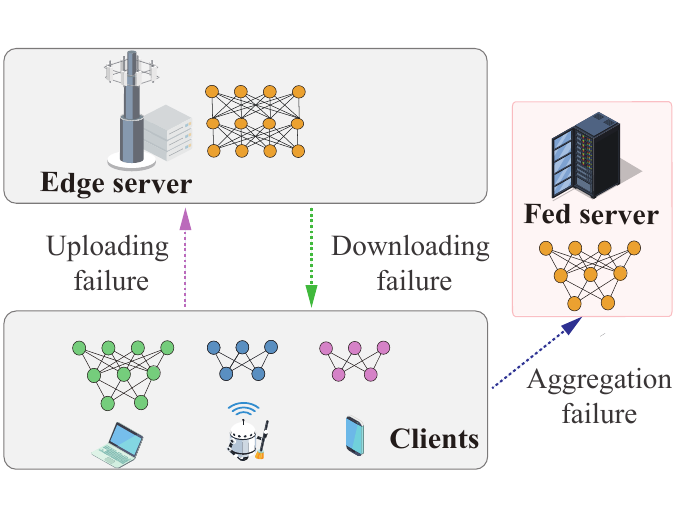}
	}
         \subfigure[\centering Test accuracy v.s. the index of cut layer.] {
 \centering
		\centering\includegraphics[width=4.1cm]{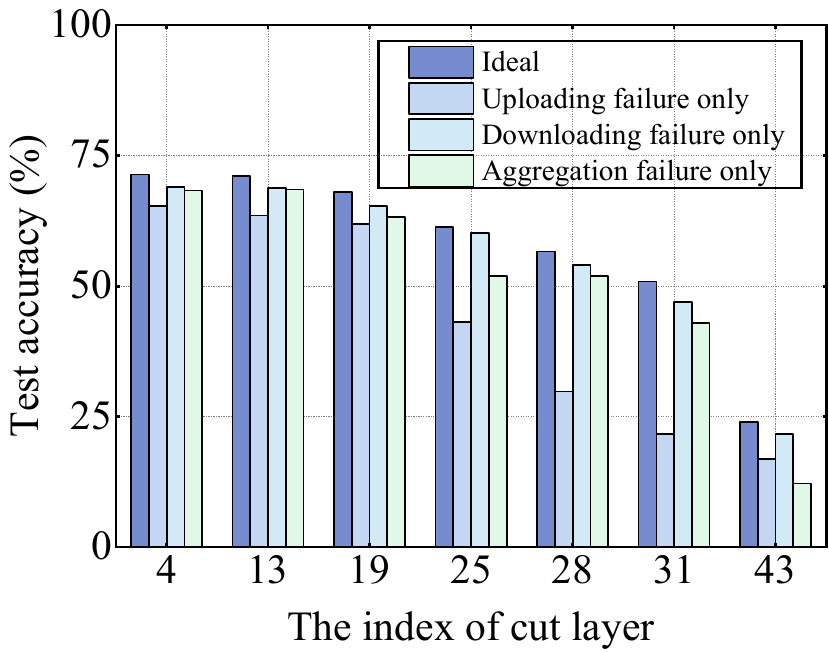}
        }
	\centering\caption{\textcolor{black}{SFL under dynamic network conditions. (a)  SFL scenario includes uploading, downloading and aggregation failures due to unstable client participation. (b) The failure occurs exclusively during uploading, downloading, or aggregation,  with a probability ranging from 0.2 to 0.6. The experiments select 5 clients out of 100 clients in each training epoch. The training process, consisting of 500 epochs, is conducted with ResNet-50 on the CIFAR-10 dataset under non-IID settings. } }
		\label{sfig:motivation_1_different_I}
        \vspace{-0.8em}
\end{figure}

However, although tremendous efforts have been made on SFL in recent years, prior SFL schemes could perform poorly under unreliable communication networks. Specifically, SFL requires not only periodical client-side model aggregation but also frequent exchanges of activations and activations' gradients between clients and the edge server during every forward and backward pass~\cite{thapa2022splitfed,lin2025hasfl,9923620,lin2025hsplitlora,chen2026siftmoesimilarityawareenergyefficientexpert,qu2026slidesimultaneousmodeldownloading}. In practice, device instability and network fluctuations, caused by memory overflow, client mobility, or deep channel fading, can easily cause updating/transmission failures in SFL. As shown in Fig.~\ref{sfig:motivation_1_different_I} (a), three major failure cases exist. (i) \textit{uploading failures} may occur when a client cannot complete the forward pass in time due to limited processing power or fails to upload activations promptly because of poor uplink channel conditions; (ii) \textit{downloading failures} arise when the downlink channel conditions from the edge server is too harsh to download activations' gradients timely; and (iii) \textit{aggregation failures} happen when a client experiences adverse uplink or downlink channel conditions when exchanging client-side models with the Fed server for aggregation. Each of these failure cases negatively affect the overall efficiency and model convergence in SFL.

Tackling unstable client participation requires rigorous convergence analysis of SFL with unstable client participation. 
The \textbf{main challenge} lies in the fact that failure can occur at different SFL stages, and these stages have \textit{varying impacts} on training performance. 
As shown in Fig.~\ref{sfig:motivation_1_different_I} (b), uploading failures prevent both server-side and client-side model updates in the affected iteration, which causes the most severe degradation. Downloading failures only halt client-side updates and lead to a smaller impact. Aggregation failures result in missing client-side models and have a moderate effect.
The combined impact of these different failure types has not been theoretically quantified in the literature. This interplay fundamentally distinguishes the convergence analysis and optimization framework for SFL from those of conventional FL systems~\cite{9060868,10292582,chen2024fedmeld,hu2024accelerating}. Moreover, while some recent SFL studies have examined various aspects to improve robustness of SFL via dynamic weight adjustments \cite{10095067} and packet loss resilience \cite{shiranthika2023splitfedresiliencepacketloss}, none of them has theoretically characterized the impact of unstable client participation on SFL and developed a unified optimization framework under unstable conditions.

\textcolor{black}{To address the above challenges, in this paper, we develop an SFL framework considering unstable client participation. To mitigate the adverse effect, optimizing client sampling and model splitting is the key. (i) \textit{Client sampling} should strike an optimal diversity-stability trade-off by sampling clients to preserve data diversity while also prioritizing relatively stable clients under non-IID data distributions. (ii) \textit{Model splitting} can also be optimized, as it affects communication-computing latency and determines how significantly client-side update failures influence overall training performance. Towards this end, we \textit{first} perform the convergence analysis of SFL under stable client participation by considering an impact of split points and failures in different stages. \textit{Then}, we devise an optimal client sampling and model splitting strategy for the SFL framework based on the theoretical upper bound. The primary contributions are as follows:}

\textcolor{black}{
\begin{itemize}
\item  We provide the first theoretical convergence analysis that reveals how model splitting and client sampling strategies impact training convergence in SFL with unstable client participation. It reveals that shallower model splitting and more frequent client-side aggregation improve convergence. The analysis provides useful engineering insights that uploading failures have a more pronounced impact on convergence than downloading or aggregation failures, and that unstable environments inevitably degrade convergence performance if not properly addressed.
 \item Based on the convergence bound, we formulate a joint optimization framework for client sampling and model splitting aimed at minimizing the upper bound of training loss. We further derive an optimal polynomial-time solution, enabling an effective balance between convergence accuracy and communication--computation efficiency in practical SFL deployments.
  \item Extensive simulations on EMNIST and CIFAR-10 under both IID and non-IID settings demonstrate that our proposed solution achieves superior test accuracy compared with baselines.
\end{itemize}}

The remainder is outlined as follows. Section~\ref{Rel_Work} reviews prior literature. Section~\ref{sec:systemModel} introduces the system model. We further provide the convergence analysis for the optimal client sampling and model splitting framework and present the problem formulation in Section~\ref{convergence_SFL}. Section~\ref{solu_appro} develops the optimal solution to this problem, and Section~\ref{simulations} presents the simulation results. Finally, in Section~\ref{conclusion}, we conclude our work.

\begin{table}[t]\label{notation}
\caption{{Summary of Important Notations.}}
{\small 
  \renewcommand{\arraystretch}{0.95}{
  \setlength{\tabcolsep}{0.1mm}{
\begin{tabular}{ll}
			\toprule
			\toprule
\textbf{Notation}                                                              & ~~~\textbf{Description}  \\ \hline
~~~$\mathcal{N}$ &  ~~~The set of all clients     \\
~~~$\mathcal{K}^{(t)}{(\bm q)}$ &  ~~~The set of clients randomly sampled in round 
$t$   \\
~~~$N$ &  ~~~The number of clients     \\
~~~$K$ &  ~~~The number of sampled clients     \\
~~~${\mathcal{D}_i}$           &  ~~~Client $i$'s local dataset  \\
~~~$f\left( {\bf{w}} \right)$     &  ~~~The global loss function parameterized by ${{\bf{w}}}$ \\
~~~$R$     &  ~~~The number of training rounds \\
~~~${\mathcal{B}_i}$     &  ~~~The mini-batch sampled from ${\mathcal{D}_i}$ \\
~~~${b}$     &  ~~~The size of mini-batch \\
~~~${{\bf{w}}_{c,i}}$/${{\bf{w}}_{s,i}}$            &  ~~~The client-side/server-side sub-model of client $i$  \\
~~~${\bf{ h}}_{s}$           &  ~~~The server-side common sub-model  \\
~~~${{\bf{ h}}_{m,i}}$           &  ~~~The server-side non-common sub-model of client $i$  \\
 ~~~${{\bf{h}}_{c,i}}$           &  ~~~The forged client-specific model of client $i$  \\
~~~$L_c$           &  ~~~The maximum cut layer index of all clients  \\
~~~$L$           &  ~~~The number of layers in the global model  \\
~~~$L_c^i$           &  ~~~The cut layer of client $i$   \\
~~~$\beta, {\sigma _j^2}, {G _j^2}$     &  ~~ Loss function constants (detailed in Section~\ref{convergence_SFL})\\
~~~$I$           &  ~~~The interval (in iterations) between client-side MA
\\
~~~$p_i$           &  ~~~Upload failure probability from client $i$
 to server
\\
~~~$\varphi_i$           &  ~~~Download failure  probability  from  server to client $i$
\\
~~~$a_i$           &  ~~~Upload failure probability from client $i$
 to Fed server
\\
~~~$\bm q$           &  ~~~Client sampling probability
\\
~~~$m_i$           &  ~~~The weight of client i's dataset
\\
~~~${s_{u}^{i}}, {s_{d}^{i}}, {s_{a}^{i}}$           &  ~~~Bernoulli random variables
\\
~~~$\gamma$           &  ~~~Learning rate
\\ \bottomrule
\end{tabular}}}
}
\end{table}

\section{Related Work}\label{Rel_Work}

\textbf{Convergence analysis of FL.}
FedAvg~\cite{mcmahan2017communication} is widely recognized as the first and the most commonly used FL algorithm. Several works have shown the convergence of FedAvg under different settings, e.g., IID setting~\cite{10.5555/3546258.3546471,NEURIPS2018_3ec27c2c}, non-IID setting~\cite{10292582,9261995} even with partial clients participation~\cite{cho2020clientselectionfederatedlearning}.
Under IID conditions, Wang et al. \cite{10.5555/3546258.3546471} present a unified, communication-efficient FL framework that reduces communication overhead and accelerates convergence with convergence guarantees.
Moreover, Woodworth et al.~\cite{NEURIPS2018_3ec27c2c} introduce a graph‑based oracle model for parallel stochastic optimization and derive the optimal lower bound that depends only on the graph depth and size, thereby clarifying fundamental limits of communication–parallelism trade-offs. For non-IID settings, Rodio et al.~\cite{10292582} analyze heterogeneous and temporally/spatially correlated client availability, demonstrating that correlation deteriorates FedAvg’s convergence rate if not handled properly. Dinh et al.~\cite{ 9261995} propose FEDL under strong convexity and smoothness, establish linear convergence by controlling local inexactness and learning rate. They derive closed-form solutions jointly tuning FL hyperparameters to balance wall-clock convergence and device energy costs. Meanwhile, Cho et al. \cite{cho2020clientselectionfederatedlearning} study biased client selection under partial participation and show that prioritizing high-loss clients accelerates FedAvg’s convergence but may introduce a small bias tied to data heterogeneity.
\textcolor{black}{Recent work has extended FL theory and system design to unreliable communication and intermittent participation scenarios, proposing robust learning mechanisms~\cite{9716792,wang2025robustfederatedlearningunreliable,10336724}. In particular, Ye et al.~\cite{9716792} study FL with unreliable communications, provide convergence analysis under stochastic link failures, and propose a Soft-DSGD approach that performs partial aggregation of received packets with reliability-aware mixing weights to improve robustness against packet losses.
Wang et al.~\cite{wang2025robustfederatedlearningunreliable} investigate robust FL in unreliable wireless networks and propose a client selection strategy that optimizes selection probabilities by considering clients’ transmission failure probabilities and local label distributions to improve training robustness and efficiency.
Sun et al.~\cite{10336724} propose stochastic coded FL, provide theoretical analysis under straggler-induced missing updates, and develop a contract-based incentive mechanism to motivate edge devices to upload less noisy coded datasets.
However, despite extensive convergence analyses in FL, existing FL convergence bounds do not apply to our context as they fail to characterize the impact of model splitting and network failures on the convergence of SFL, where training proceeds through entirely  three different stages.}



\textbf{Client sampling.} Client sampling is a critical design component in distributed machine learning across heterogeneous devices. In the existing FL literature, client sampling strategies primarily include uniform sampling \cite{mcmahan2017communication}, importance-aware sampling~\cite{chen2022optimalclientsamplingfederated,9904868,pmlr-v151-jee-cho22a}, clustering-based sampling~\cite{fraboni2021clustered,NEURIPS2024_7886b9ba}, resource-aware sampling \cite{10.1109/INFOCOM48880.2022.9796935}, and fairness-aware sampling \cite{9810502}. Uniform sampling refers to randomly selecting a fixed proportion of clients from the user pool during each training round. Due to its simplicity, uniform sampling does not account for system or data heterogeneity. Importance-aware sampling mitigates the high variance inherent in stochastic optimization methods like stochastic gradient descent (SGD) by preferentially sampling clients with more ``important'' data. Specifically, these approaches assign higher sampling probabilities to clients whose contributions are quantified using local gradients (e.g., \cite{chen2022optimalclientsamplingfederated,9904868}) or local losses (e.g., \cite{pmlr-v151-jee-cho22a}). 
Clustering-based sampling first groups clients according to the similarity of their data features or model updates and then selects representative clients from each group~\cite{fraboni2021clustered,NEURIPS2024_7886b9ba}. Resource-aware sampling takes into account the available system resources, such as communication and computing resources, when selecting participating clients~\cite{10.1109/INFOCOM48880.2022.9796935,9810502}. 
However, the above methods are tailored for traditional FL, which may be ineffective for SFL with unstable clients.

\textbf{Model splitting.} Model splitting plays a critical role in SFL, especially in scenarios with unstable client participation. The selection of the split point influences both communication/computation latency and overall training performance~\cite{ZW2024ultra-LoLa,zw2025AIoutage}. Several studies have investigated how to determine the optimal model splitting point for SFL \cite{lin2024adaptsfl,10980018,10304624,shiranthika2023splitfedresiliencepacketloss}. 
Shiranthika et al. \cite{shiranthika2023splitfedresiliencepacketloss}  study the impact of model splitting strategies on the packet loss resilience
of SFL. 
Lin et al. \cite{lin2024adaptsfl,10980018} investigate how the model splitting point affects training convergence, revealing that a shallower split point leads to smaller convergence upper bound. Xu et al. \cite{10304624} propose a combined optimization approach to identify the best split point and bandwidth allocation, aiming to reduce total training latency. Unfortunately, these studies overlook how model splitting affects training convergence with unstable client participation.

\textcolor{black}{To the best of our knowledge, this work provides the first convergence analysis of SFL that comprehensively characterizes the three failure modes (activation uploading, gradient downloading, and model aggregation) and leverages the derived bound for system optimization.}



\section{%
Preliminaries and System Model}
\label{sec:systemModel}

\begin{figure*}[t!]
\centering
\includegraphics[width=15cm]{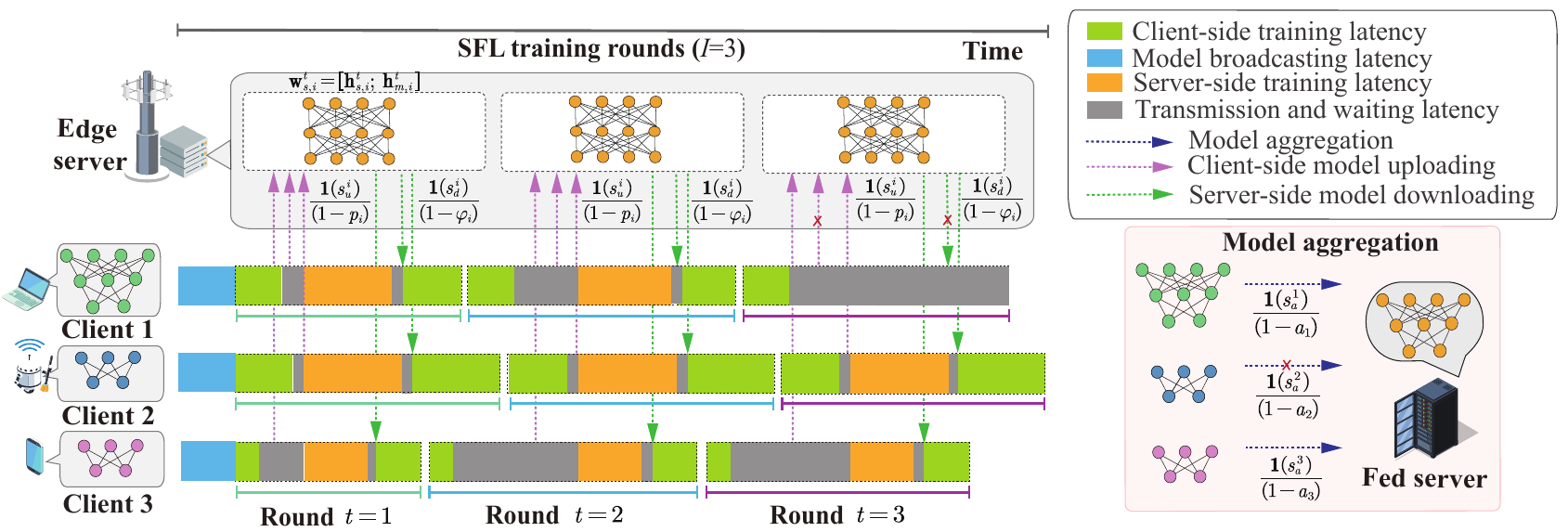}
	\flushleft
\caption{Split federated
learning with unstable clients.  The main challenge for tackling unstable client participation lies in that failure can occur at different SFL stages, including activation uploading, gradient downloading, and client-side model aggregation, which have varying impacts on training performance.  $
p_i
$
is the probability that  client $i$'s upload to the edge server fails. $\varphi_i$ represents the probability that client $i$ experiences a downlink failure from the edge server. $a_i$ is the probability that client $i$'s upload to the Fed server fails. }
\label{multi-system}
\vspace{-0.8em}
\end{figure*}
\subsection{The SFL Framework}%
We study an SFL framework with an edge server coordinating a client set $\mathcal{N} = \{1, \ldots, N\}$, where $N$ denotes the total number of clients. The overall process is described in detail as follows:
\begin{itemize}
\item \textbf{Client side}:  Each client $i$ possesses $D_i$ local training samples, where \(\mathbf{x}_{i,k}\) indicates the input data of the \(k-\)th sample and \(y_{i,k}\) represents its corresponding label, forming the dataset $\mathcal{D}_i=\{( \mathbf{x}_{i, k},y_{i, k})\}_{k=1}^{D_i}$. Across all 
$N$ clients, the total number of training samples is $D :=\sum\nolimits_{i = 1}^N D_i$. 
Moreover, ${{\bf{w}}_{c,i}}$ denotes the client-side sub-model of client $i$. 
\item \textbf{Edge server}: The edge server acts as a high-performance central computing node that performs training on server-side sub-models. For client $i$, the server-side sub-model is represented as  ${{\bf{w}}_{s,i}} = \left[{{\bf{h}}_{s}};{{\bf{h}}_{m,i}}\right]$, where ${\bf{h}}_{s}$ is the common component shared across all clients and synchronized in every round. Moreover, ${\bf{h}}_{m,i}$ is the non-common (personalized) portion, which is specific to clients with additional layers on the server and is aggregated every 
$I$ rounds. Additionally, the edge server gathers key network information, including client's computing capabilities and channel conditions, to facilitate optimized decision-making.
\item \textbf{Fed server}: 
The Fed server coordinates and synchronizes client-side sub-models by aggregating these sub-models at regular intervals from participating clients.
To address privacy risks, the Fed server and edge server are typically operated by different entities~\cite{10.1145/3460120.3485259}. 
\end{itemize}


\subsubsection{Global Model and Objective}
The global model is denoted as ${\bf{w}} = \left[ {{{\bf{w}}_{s,i}};{{\bf{w}}_{c,i}}} \right] $, where ${{\bf{w}}_{s,i}}$ represents server-side sub-model and ${{\bf{w}}_{c,i}}$ represents client-side submodel of client $i$, respectively. The objective of SFL is to minimize the global training loss
\begin{align}\label{minimiaze_loss_function}
\mathop {\min }\limits_{\bf{w}} f\left( {\bf{w}} \right) \buildrel \Delta \over = \mathop {\min }\limits_{\bf{w}} {\frac{{{1}}}{N}} \sum\limits_{i = 1}^N {f_i}({\bf{w}}),
\end{align}
where ${f_i}\left( {\bf{w}} \right) \buildrel \Delta \over = {\mathbb{E}_{{\xi _i} \sim {\mathcal{D}_i}}}[{F_i}\left( {{\bf{w}};{\xi _i}} \right)]$ is client $i$'s local loss function, and $\xi _i$ represents stochastic mini-batch sampling from dataset $\mathcal{D}_i$. Moreover, the stochastic gradient $\nabla F_{i}(\mathbf{w}^{t};\xi_{i}^{t}) $ is an unbiased estimate of $\nabla f_{i}(\mathbf{w})$ with $\mathbb{E}_{\xi_{i}^{t}\sim \mathcal{D}_{i}}[\nabla_{\bf{w}} F_{i}(\mathbf{w}^{t-1}_{i};\xi_{i}^{t}) \vert \boldsymbol{\xi}^{[t-1]}] = \nabla f_{i}(\mathbf{w}^{t-1}_{i})$, where $\boldsymbol{\xi}^{[t-1]}$ encompasses all randomness up to round $t-1$, defined as $\boldsymbol{\xi}^{[t-1]} \overset{\Delta}{=} [\xi_{i}^{\tau}]_{i\in\{1,2,\ldots,N\}, \tau\in\{1,\ldots,t-1\}}$~\cite{pmlr-v119-karimireddy20a}. Besides, {${\nabla _{\bf{w}}}F({\bf{w}}; \xi)$ denotes the first-order derivative (gradient) of $F({\bf{w}}; \xi)$ with respect to the parameter vector ${\bf{w}}$.}

\subsubsection{{Client Sampling}}%
\textcolor{black}{We sample clients based on a probability distribution
\begin{align}
\boldsymbol{q} = \{q_i\}_{i\in\mathcal{N}}, 
\quad 0 \le q_i \le 1,
\end{align}
with the constraint $\sum_{i=1}^Nq_i\!=\!1$. By optimizing $\boldsymbol{q}$, we jointly account for system‐level heterogeneity (e.g., unstable communication status) and statistical heterogeneity (non‐i.i.d. data distributions across clients) in order to minimize the global training loss for convergence. 
Following prior works 
\cite{li2019convergence,yang2021achieving}, we assume that, at each round \(t\), the server independently draws $K$ clients with replacement based on the probability \(\boldsymbol{q}\). The resulting client set is denoted as $\mathcal{K}^{(t)}(\boldsymbol{q})$,   where $\mathcal{K}^{(t)}\! \subseteq\!  \mathcal{N}$ and $K\! :=\! \left| \mathcal{K}^{(t)}\right|$.
If a given client appears multiple times in \(\mathcal{K}^{(t)}\), its update is added repeatedly  during aggregation according to its number of appearances.}

\subsubsection{Training Process}
Before training, the edge server initializes the machine learning model, splits it into client-side and server-side sub-models, distribute these sub-models to clients and edge server, and determines the optimal client sampling strategy. Afterwards, SFL training process is executed in $I$ rounds, repeating until convergence.
The workflow comprises two phases: \textit{split training} and \textit{client-side model aggregation}.  Split training occurs every round, while the model aggregation of clients is performed every $I$ rounds. Specifically, for  round $t \in \mathcal{R} = \left\{ {1,2,...,R} \right\}$, the SFL framework proceeds as follows.


\textit{a. Split Training Stage:} Executed in each training round, this stage updates sub-models of the participating clients and the edge server, and is composed of three sequential steps:

\textit{a1) Client-side Model Forward Pass and Activations' Transmission:} In this step, the selected clients execute client-side forward pass (FP) simultaneously. To be specific, client $i$ samples a mini-batch ${\mathcal{B}_i}$ of batch size
$b $ from the local dataset ${\mathcal{D}_i}$, 
feeds it through the corresponding client-side sub-model to generate activations at the cut layer, and then transmits these activations along with the labels to the edge server, typically via wireless channels. The edge server then uses the received activations to update the server-side sub-model.

\textit{a2) Server-side Model Forward Pass and Backward Pass:} The edge server receives the activations and feeds them into the respective server-side sub-models. Specifically, for client $i$, the prediction is given by
\begin{align}\label{stage_1_3}
{\bf{\hat y}}^t_i = \eta  \left( {{\bf{a}}^t_i;{{\bf{w}}^{t-1}_{s,i}}} \right),
\end{align}
where ${{\bf{w}}^{t-1}_{s,i}} = \left[{{\bf{ h}}^{t-1}_{s}};{{\bf{ h}}^{t-1}_{m,i}}\right]$. Here,
${\bf{ h}}^{t-1}_{s}$ represent the server-side common sub-model, while ${\bf{h}}^{t-1}_{m,i}$ refers to server-side non-common sub-model. 
Specifically, the edge server updates $\mathbf{h}_{s}^{t}$ as follows
\begin{align}\label{stage_5_2}
\mathbf{h}_{s}^{t}=\frac{1}{K}\sum_{i\in \mathcal{K} ^{(t)}(\boldsymbol{q})}{\frac{m_i}{q_i}}\mathbf{h}_{s,i}^{t},
\end{align}
where $m_i$ represents client $i$'s dataset weight. Moreover,  $\mathbf{h}_{s,i}^{t}$ denotes the server-side common sub-model  of client $i$, which is detailed by
\begin{align}
\,\,\mathbf{h}_{s,i}^{t}\gets \mathbf{h}_{s,i}^{t-1}-\gamma \frac{\mathbf{1}\!\left( {s_{u}^{i}} \right)}{\left( 1-p_i \right)}\nabla _{\mathbf{h}_s}F_i(\mathbf{h}_{s,i}^{t-1};\xi _{i}^{t}),
\end{align}
where $\gamma$ denotes the learning rate and $\nabla _{\mathbf{h}_s}F_i(\mathbf{h}_{s,i}^{t-1};\xi _{i}^{t})$ represents the stochastic gradient.  
Here, ${s_{u}^{i}}$ is the Bernoulli random variable that indicates whether the activations upload of the client $i$ to the server is successful \cite{10121038}.
Specifically,
$
{s_{u}^{i}} \sim \mathrm{Bernoulli}(1 - p_i),
$
where 
$
p_i
$
is the probability that  client $i$'s upload to the edge server fails.
Moreover, the indicator function expressed by
\begin{align}
&\mathbf{1}\!\left( {s_{u}^{i}} \right) =\left\{ \begin{array}{l}
	1,\ \text{with}\ \text{probability}\ 1-p_i,\\
	0,\ \text{with}\ \text{probability}\ p_i,\\
\end{array} \right. 
\end{align}
ensures that only the updates from clients whose uploads have actually succeeded are included in the subsequent server-side aggregation, while dividing by \(1-p_i\) keeps the estimator unbiased in expectation despite random upload failures.
The server-side non-common sub-model $\mathbf{h}_{m,i}^{t}$ is dictated by the heterogeneous cut layers: 
clients that offload more layers to the server induce personalized parts of server-side sub-models. Moreover, the updating rule for each client \(i\) updates its own non‐common sub‐model is
\begin{align}\label{non_commen_update}
\mathbf{h}_{m,i}^{t}\gets \mathbf{h}_{m,i}^{t-1}-\gamma \frac{\mathbf{1}\!\left( {s_{u}^{i}} \right)}{\left( 1-p_i \right)}\nabla _{\mathbf{h}_m}F_i(\mathbf{h}_{m,i}^{t-1};\xi _{i}^{t}),
\end{align}
where  ${\nabla _{{{\bf{h}}_m}}}{F_i}({\bf{h}}_{m,i}^{t - 1};\xi _i^t)$ represents the stochastic gradient. 


\textit{a3) Activations' Gradients Transmissions and Client-side Model Backward Pass:} After the server‐side backward pass finishing, the edge server transmits gradients back to each client. Upon receiving these gradients, each client updates its client‐side sub‐model. For client \(i\), the update rule is
\begin{align}\label{client_side_update}
\mathbf{w}_{c,i}^{t}\gets \mathbf{w}_{c,i}^{t-1}-\gamma \frac{\mathbf{1}\!\left( {s_{u}^{i}},\,\,{s_{d}^{i}} \right)}{\left( 1-p_i \right) \left( 1-\varphi _i \right)}\nabla _{\mathbf{w}_c}F_i(\mathbf{w}_{c,i}^{t-1};\xi _{i}^{t}),
\end{align}
where $
{s_{u}^{i}} \sim \mathrm{Bernoulli}(1 - p_i)$ and 
${s_{d}^{i}} \sim \mathrm{Bernoulli}(1 - \varphi_i)
$.
Moreover, 
$\varphi_i$ represents the probability that client $i$ experiences a downlink failure from the edge server. 
The combined indicator
satisfies
\begin{align}
\Pr\{\,\mathbf{1}({s_{u}^{i}},{s_{d}^{i}})=1\} \;=\;(1-p_i)\,(1-\varphi_i).
\end{align}
Similarly, the indicator function in Eq.~\eqref{client_side_update} divided by \((1-p_i)(1-\varphi_i)\) ensures that the estimator remains unbiased in expectation despite possible communication failures.


\textit{b. Client-side Model Aggregation Stage:}  gathers and aggregates the client-specific models at the Fed server.
Each client-specific model consists of $\mathbf{h}_{m,i}^{t}$ and $\mathbf{w}_{c,i}^{t} $, aggregates every $I$ training rounds. The stage proceeds in three steps:

  \textit{b1) Model Upload:} 
Client  \(i\) and the edge server send
    the updated client‐side sub‐model \(\mathbf{w}_{c,i}^t\) and \(\mathbf{h}_{m,i}^t\) to the Fed server over the wireless/wired link.  Given the highly reliable inter-server link between edge server and the Fed server with sufficient bandwidth, we focus  on the upload failure of \(\mathbf{w}_{c,i}^t\) from clients to the Fed server in our analysis, and $a_i$ is the probability that client $i$'s upload to the Fed server fails.

  \textit{b2) Client‐Side Model Aggregation:}
    Upon receipt of ${{\bf{h}}^t_{m,i}}$ and ${{\bf{ w}}^t_{c,i}}$, the Fed server ``forges" each client’s full model by concatenating its server‐side and client‐side parts:
\begin{align}
\Delta {\bf{h}}_{c,i}^{t}=\left[ {\bf{h}}_{m,i}^{t}-{\bf{h}}_{m}^{t-I}\,\,; \frac{\mathbf{1}(s_{a}^{i})}{1-a_i}\left( {\bf{ w}}_{c,i}^{t}-{\bf{ w}}_{c}^{t-I} \right) \right], 
\label{aggre-client}
    \end{align}
    where ${s_{a}^{i}} \sim \mathrm{Bernoulli}(1 - a_i)$. Here, \(\mathbf{1}({s_{a}^{i}})=1\) if the assembly succeeds (with probability \(1-a_i\)), and \(0\) otherwise.  
    The Fed server then aggregates these forged models as
    \begin{align}\label{h_c_define_refined}
    \mathbf{h}_c^t = \mathbf{h}_c^{t-I} + \frac{1}{K} \sum_{i \in \mathcal{K}^{(t)}(\bm q)} \frac{m_i}{q_i} \Delta \mathbf{h}_{c,i}^t,
    \end{align}
    where \(\mathcal{K}^{(t)}(\boldsymbol{q})\) refers to the set of \(K\) clients sampled with replacement according to \(\boldsymbol{q}\).

      \textit{b3) Model Download:} 
    Finally, the updated client-side sub-models are delivered to the corresponding clients, and the updated server-side non-common sub-models are deployed on the edge server, respectively.


This process is repeated until the global training loss converges. Although recent studies have established convergence guarantees for FL with client sampling \cite{li2019convergence,yang2021achieving}, they typically adopt uniform or data‐size‐proportional client sampling. 
Likewise, convergence analyses for SFL \cite{lin2024adaptsfl,han2025convergenceanalysissplitfederated} assume idealized conditions, such as stable network links and consistently available clients, conditions that rarely hold in realistic deployments.
 When deployed in unstable environments—characterized by fluctuating network links' conditions, intermittent client availability, and heterogeneous resources—these sampling strategies aggravate straggler delays, thereby prolonging total convergence time, and amplify oscillations in the aggregated model, resulting in degraded training accuracy.
  Thus, an optimal client sampling strategy that jointly accounts for unstable client participation and statistical heterogeneity is crucial for maintaining high accuracy while keeping the total training delay within acceptable bounds.

\begin{algorithm}[t]
\begin{spacing}{0.8} 
	\renewcommand{\algorithmicrequire}{\textbf{Input:}}
	\renewcommand{\algorithmicensure}{\textbf{Output:}}
    	\caption{Training Framework for SFL in Unstable Environments.}\label{UnstableSFL}
	\begin{algorithmic}[1]
 \REQUIRE    mini-batch size $b$, learning rate $\gamma$, epochs $E$, the client set ${\cal N}$, data set $\cal{D}$, the number of participating clients $K$, initial client-side models $\{\mathbf{w}_{c,i}^{0}\}_{i\in\mathcal{N}}$, initial server-side models $\{[\mathbf{h}_{s,i}^{0},\mathbf{h}_{m,i}^{0}]\}_{i\in\mathcal{N}}$, uplink data rates \(\{r_i^U\}_{i\in\mathcal N}\),  
  downlink data rates \(\{r_i^D\}_{i\in\mathcal N}\),  
  the Fed server's uplink data rates \(\{r_i^A\}_{i\in\mathcal N}\).
		\ENSURE ${{{{\bf{w}}^E}}^{{*}}}$ and $\{[{\mathbf{h}_{s,i}^{E}}^*,{\mathbf{h}_{m,i}^{E}}^*]\}_{i\in\mathcal{N}}$. 
                  \STATE Obtain $\bm q$ and $L_c^i$ based on~\textbf{Algorithm~\ref{alg:short}}.
          \WHILE{$\rho < E$}
          \STATE \(I_\tau\leftarrow I_{\text{global}}\),\quad \(I_{\mathrm{eff}}\leftarrow\min(I_\tau,\,E-\rho)\)
           \STATE $\mathcal{K}$ $\leftarrow$  Select $K$ clients from  ${\cal N}$ based on  $\bm q$.
  \FOR{\(t=\rho+1,\dots,\rho+I_{\mathrm{eff}}\)}
            \STATE \textbf{// Client-side execution  }
          \FOR {client ${i \in \,{\mathcal{K}}}$ simultaneously}
            \STATE  ${{\bf{a}}^t_i} \leftarrow \varphi \left( {{\bf{x}}^t_i};{{\bf{w}}^{t-1}_{c,i}} \right)$
              \IF{${s_{u}^{i}}=1$}
\STATE Send $\left( {{{\bf{a}}^t_i},{\bf{y}}_i^t} \right)$ to the edge server.
  \ENDIF
          \ENDFOR
          \STATE
            \STATE \textbf{// Executes on the edge server  }
          \STATE${\bf{\hat y}}^t_i = \varphi\left( {{\bf{a}}^t_i;{{\bf{w}}^{t-1}_{s,i}}} \right)$
          \STATE Calculate loss function value $f\left( {{\bf{w}}^{t - 1}} \right)$
          \STATE $
\mathbf{h}_{s}^{t}\gets \mathbf{h}_{s}^{(t-1)}+\sum\limits_{i\in \mathcal{K} }{\frac{1}{K}}\frac{m_i}{q_i}\left( \mathbf{h}_{s,i}^{t}-\mathbf{h}_{s}^{(t-1)} \right) 
$
          \STATE $
\mathbf{h}_{m,i}^{t}\gets \mathbf{h}_{m,i}^{t-1}-\gamma \frac{\mathbf{1}\!\left( {s_{u}^{i}} \right)}{\left( 1-p_i \right)}\nabla _{\mathbf{h}_m}F_i(\mathbf{h}_{m,i}^{t-1};\xi _{i}^{t})
$
              \IF{${s_{d}^{i}}=1$}
\STATE Send  gradients  to corresponding clients.
  \ENDIF

           \STATE
            \STATE \textbf{// Client-side execution  }
         \FOR {client ${i \in \,{\cal N}}$ simultaneously}
           \STATE $
\mathbf{w}_{c,i}^{t}\gets \mathbf{w}_{c,i}^{t-1}-\gamma \frac{\mathbf{1}\!\left( {s_{u}^{i}},\,\,{s_{d}^{i}} \right)}{\left( 1-p_i \right) \left( 1-\varphi _i \right)}\nabla _{\mathbf{w}_c}F_i(\mathbf{w}_{c,i}^{t-1};\xi _{i}^{t})
$
  \IF{${s_{a}^{i}}=1$ and $\left( {t - \rho } \right)$ mod $I^\tau$ $=0$}
\STATE Send $\mathbf{w}_{c,i}^{t}$  to fed server.
  \ENDIF

        \ENDFOR

         \STATE
                    \STATE \textbf{// Executes on the fed server  }
        \IF{$\left( {t - \rho } \right)$ mod $I^\tau$ $=0$}
          \STATE Forge client-side specific sub-models according to Eq.~\eqref{aggre-client}

          \STATE $
\mathbf{h}_c^t = \mathbf{h}_c^{t-I} + \frac{1}{K} \sum_{i \in \mathcal{K}^{(t)}(\bm q)} \frac{m_i}{q_i} \Delta \mathbf{h}_{c,i}^t
$

          \STATE ${\bf{h}}_{c,i}^t \leftarrow {\bf{h}}_c^t$
        \ENDIF
        \ENDFOR
        \STATE $\rho  \leftarrow \rho+I_{\mathrm{eff}}$, $\tau \leftarrow \tau+1$
        
          \ENDWHILE          
	\end{algorithmic}  
    \end{spacing}
\end{algorithm}
\section{Convergence Analysis}\label{convergence_SFL}
In this section, we present the first convergence analysis for unstable SFL, examining how client sampling strategies and model splitting jointly influence training convergence, which motivates developing an efficient iterative method in Section~\ref{solu_appro}. Furthermore, we present the problem formulation for the client sampling problem in SFL, addressing the unstable client participation issue.
\subsection{Convergence Analysis for the Optimal  Client Sampling and Model Splitting }
We denote client $i$'s global model at round $t$ by $
\mathbf{w}_{i}^{t}=\left[ \mathbf{h}_{s,i}^{t}; \mathbf{h}_{c,i}^{t} \right]. 
$
 Specifically, $\mathbf{h}_{s,i}^{t}$ are aggregated every training rounds, while $\mathbf{h}_{c,i}^{t}$ is aggregated every $I$ training round. The cut layer $L_c$ between $\mathbf{h}_{s,i}^{t}$ and $\mathbf{h}_{c,i}^{t}$ is specified by the maximal depth of the client-side sub-model across all clients, accommodating heterogeneity in individual split locations. The corresponding gradients for $\mathbf{h}_{s,i}^{t}$ and $\mathbf{h}_{c,i}^{t}$ are specified as follows:
\begin{equation}\label{g_ci}
\mathbf{g}_{c,i}^{t}=\left[ \nabla _{\mathbf{h}_m}F_i(\mathbf{h}_{m,i}^{t-1};\xi _{i}^{t}); \nabla _{\mathbf{w}_c}F_i(\mathbf{w}_{c,i}^{t-1};\xi _{i}^{t}) \right],
\end{equation}
and 
\begin{equation}
{\bf{g}}_{s,i}^t = {{\nabla _{{{\bf{h}}_s}}}{F_i}({\bf{h}}_{s}^{t - 1};\xi _i^t)},
\end{equation}
where $
\mathbf{h}_{s}^{t}=\frac{1}{K}\sum_{i\in \mathcal{K} ^{(t)}(\boldsymbol{q})}{\frac{m_i}{q_i}}\mathbf{h}_{s,i}^{t}
$. Moreover, the corresponding gradient is denoted by ${\bf{g}}_i^t = \left[ {{\bf{g}}_{s,i}^t; {\bf{g}}_{c,i}^t} \right]$.

\textcolor{black}{Following prior work in distributed stochastic optimization \cite{9261995,pmlr-v162-gao22c,yang2025optBS},
which focuses on analyzing the convergence of an aggregated form for individual solutions, we analyze the convergence of $
\mathbf{w}^t=\frac{1}{K}\sum_{i\in \mathcal{K} ^{(t)}(\boldsymbol{q})}{\frac{m_i}{q_i}}\mathbf{w}_{i}^{t}$. 
To further evaluate the convergence upper bound of \textbf{Algorithm~\ref{UnstableSFL}}, we adopt two common assumptions regarding the training loss functions,  which are consistent with
distributed machine learning literature \cite{karimireddy2020scaffold,zhang2013communication,mania2017perturbed}.}

\begin{assumption}[\textit{Smoothness}]\label{asp:1}
\textit{Each local loss function ${f_i}\left( {\bf{w}} \right)$ is differentiable and $\beta $-smooth, i.e., for all $\mathbf{w}$ and ${{\bf{w'}}}$, }
    \begin{equation}
\left\| {\nabla_{\bf{w}} {f_i}\left( {\bf{w}} \right) - \nabla_{\bf{w}} {f_i}\left( {\bf{w'}} \right)} \right\| \le \beta \left\| {{\bf{w}} - {\bf{w'}}} \right\|,\;\forall i.
    \end{equation}
\end{assumption}

\begin{assumption}[\textit{Bounded variances and second moments}]\label{asp:2}
\textit{The variance and second moments of stochastic gradients for each layer have upper bound}
    \begin{equation}
\mathbb{E}_{\xi_{i}\sim \mathcal{D}_{i}} \Vert \nabla_{\bf{w}} F_{i}(\mathbf{w}; \xi_{i}) - \nabla_{\bf{w}} f_{i}(\mathbf{w})\Vert^{2} \leq  \sum\limits_{j = 1}^l  {\sigma _j^2},\ \forall \mathbf{w}, \; \forall i,
    \end{equation}
        \begin{equation}
\mathbb{E}_{\xi_{i}\sim \mathcal{D}_{i}} \Vert \nabla_{\bf{w}} F_{i}(\mathbf{w}; \xi_{i}) \Vert^{2} \leq \sum\limits_{j = 1}^l  {G _j^2},\ \forall \mathbf{w}, \;\forall i ,
    \end{equation}
\textit{where $l$ is number of layers for model  $\mathbf{w}$, and ${\sigma _j^2}$ and ${G_j^2}$ bound the gradient variance and second moment  for the $j$-th layer of model $\bf w$, respectively.}
\end{assumption}



\begin{lemma} \label{lm:diff-avg-per-node}
Under {\bf Assumption \ref{asp:1}} and using {\bf Algorithm \ref{UnstableSFL}}, the following bound holds for the discrepancy between the aggregated client model $\mathbf{h}_c^t$ and an individual client’s full (forged) model $\mathbf{h}_{c,i}^t$:
\begin{align}
&\mathbb{E} [\Vert {\mathbf{h}}_c^{t} - \mathbf{h}^{t}_{c,i}\Vert^{2}] \leq 
2\gamma ^2I^2\bigg(N\underset{i^\prime}{\max}\big\{ \frac{{m_{i^\prime}}^2}{q_{i^\prime}}\frac{1}{\left( 1-a_{i^\prime} \right)}\frac{1}{\left( 1-p_{i^\prime} \right) }\notag \\
&\quad \quad \frac{1}{ \left( 1-\varphi_{i^\prime} \right)} \big\} +\frac{1}{\left( 1-a_i \right)}\frac{1}{\left( 1-p_i \right) \left( 1-\varphi _i \right)}\bigg)\sum_{j=1}^{L_c}{G_{j}^{2}}
, \forall t,
\end{align}
\textit{where $I$ denotes the frequency of client-side model aggregation.}
\end{lemma}

\begin{proof}
See Appendix A.
\end{proof}

\begin{theorem}\label{theorem1}
With {\bf Assumption \ref{asp:1}}, {\bf Assumption \ref{asp:2}}, and  {\bf Lemma \ref{lm:diff-avg-per-node}}, the following convergence bound holds for all $R \geq 1$:
{ \small \begin{equation}\label{convergence_bound}
\begin{aligned}
&\frac{1}{R}\sum\limits_{t = 1}^R \mathbb{E}  [{\Vert\nabla _{\bf{w}}}f({{\bf{w}}^{t - 1}})\Vert{^2}] 
\leq
  \frac{2\vartheta }{\gamma R} 
+\beta\gamma \sum_{i=1}^N\frac{{m_i}^2}{q_i}\cdot
 \\
&
\frac{1}{\left( 1-p_i \right) }\bigg\{\frac{1}{\left( 1-\varphi _i \right)}\frac{1}{\left( 1-a_i \right)}\sum_{j=1}^{L_c}({\sigma _{j}^{2}}+{G_j^2})+\sum_{j=L_c+1}^L({\sigma _{j}^{2}}+
{G_j^2})\bigg\}
 \\
&+\sum_{i=1}^N{}\beta ^2\frac{{m_i}^2}{q_i}2\gamma ^2I^2 \big( N\underset{i^\prime\in\![1,N]}{\max}\big\{ \frac{{m_{i^\prime}}^2}{q_{i^\prime}}\frac{1}{\left( 1-p_{i^\prime} \right) } \frac{1}{ \left( 1-\varphi_{i^\prime} \right)}\frac{1}{\left( 1-a_{i^\prime} \right)}\big\} \\
&+\frac{1}{\left( 1-a_i \right)}\frac{1}{\left( 1-p_i \right) \left( 1-\varphi _i \right)} \big) \sum_{j=1}^{L_c}{G_{j}^{2}}\triangleq \mathcal{U}\left(\boldsymbol q_{L_c},\{L_c^i\}_{i=1}^N\right),
\end{aligned}
\end{equation}}%
where $\vartheta = f(\mathbf{w}^0) - f^{*}$, with $L$ denoting the total number of layers in the global model and $f^{*}$ the optimal objective value of \textbf{Problem~P1}.
\end{theorem}

\begin{proof}
See Appendix B.
\end{proof}

\begin{corollary}\label{corollary1}
The number of {training rounds $R$} required to achieve target convergence accuracy $\varepsilon$ is bounded as follows
\begin{equation}
\frac{1}{R} \sum_{t=1}^{R} \mathbb{E} [\Vert \nabla_{\bf{w}} f({\mathbf{w}}^{t-1})\Vert^{2}] \le \varepsilon,
\end{equation}
which can be further expressed by
\begin{equation}\label{accuracy_cons_corollary}
\frac{2\vartheta}{\gamma \left( \varepsilon +\varGamma \right)}\le R
,
\end{equation}
where
{ \small 
\begin{align}
&\varGamma = -\beta\gamma \sum_{i=1}^N\frac{{m_i}^2}{q_i}\frac{1}{\left( 1-p_i \right) }\bigg\{\frac{1}{\left( 1-\varphi _i \right)}\cdot
 \notag\\
&
\frac{1}{\left( 1-a_i \right)}\sum_{j=1}^{L_c}({\sigma _{j}^{2}}+{G_j^2})+\sum_{j=L_c+1}^L({\sigma _{j}^{2}}+
{G_j^2})\bigg\}-\sum_{i=1}^N{}\beta ^2\frac{{m_i}^2}{q_i}
 \notag\\
& 2\gamma ^2I^2 \cdot\big( N\underset{i^\prime\in\![1,N]}{\max}\big\{ \frac{{m_{i^\prime}}^2}{q_{i^\prime}}\frac{1}{\left( 1-p_{i^\prime} \right) } \frac{1}{ \left( 1-\varphi_{i^\prime} \right)}\frac{1}{\left( 1-a_{i^\prime} \right)}\big\} \notag\\
&+\frac{1}{\left( 1-a_i \right)}\frac{1}{\left( 1-p_i \right) \left( 1-\varphi _i \right)} \big) \sum_{j=1}^{L_c}{G_{j}^{2}}
.
\end{align}
}
\end{corollary}

\textbf{Insight 1:}  Eq.~\eqref{convergence_bound} and Eq.~\eqref{accuracy_cons_corollary} show that for a given {training round $R$}, increasing the frequency of client-side aggregation and adopting a shallower cut layer improve converged accuracy (i.e., reduce $\varepsilon$). 
Furthermore, as $p_i$, $\varphi_i$ and $a_i$ increase, achieving target convergence accuracy $\varepsilon$ also requires more training rounds, indicating that unstable environments negatively impact the global training loss.
By optimally selecting the split point $L_c$ and adopting a client sampling policy that prioritizes clients with more stable uplink and downlink connections to edge server/fed server, it can minimize the global training loss within the same budget of rounds.  Thus, jointly optimizing model splitting and client sampling is essential for efficient SFL under unstable environments.

\textbf{Insight 2:} The convergence bound reveals that \(p_i\) has a relatively significant impact on the convergence compared with \(\varphi_i\) or \(a_i\), as it amplifies the contributions of both error and variance terms. Even slight increases in \(p_i\) can lead to a significant increase in the bound. The result matches our intuition and the phenomenon in Fig.~\ref{sfig:motivation_1_different_I}, as uploading failure interrupts both client-side and server-side model updates.

\textbf{Insight 3:} When $\varphi_i$ and $a_i$ are relatively high, choosing a smaller $L_c$ is beneficial. A shallower splitting point results in a larger proportion of server-side model. In this way, even if failures occur during the gradient downloading process from the server to the client or during the client-side model aggregation at the Fed server, a smaller $L_c$ can help mitigate the effect of discarding model updates.

\subsection{Problem Formulation} 
Our aim is to jointly optimize client selection probabilities $\bm q$ and model split points \(\{L_c^i\}_{i=1}^N\) to accelerate convergence while satisfying delay and resource constraints. In gradient‐based methods~\cite{lin2024adaptsfl,10980018}, driving the squared norm of the gradient toward zero indicates that the algorithm is approaching a stationary point. Hence, we aim to minimize the average expected squared norm of the gradients $\frac{1}{R}\sum\limits_{t=1}^R
\mathbb{E} [{\Vert\nabla _{\bf{w}}}f({{\bf{w}}^{t}}(\boldsymbol{q},L_{c}^{ i }))\Vert{^2}]$ over 
$R$ rounds, which motivates the following problem:
\begin{align}
\label{ob1}
\!\!\!\!\!\!\!\textbf{P1}:\quad &\underset{\{ \boldsymbol{q},\ L_{c}^{ i }
\}} \min  \quad 
\dfrac{1}{R}\sum\limits_{t=1}^R
\mathbb{E} [{\Vert\nabla _{\bf{w}}}f({{\bf{w}}^{t}}(\boldsymbol{q},L_{c}^{ i }))\Vert{^2}] \notag\\
\quad\quad \text { s.t. } 
&~\mathrm{C1:}~
\sum\limits_{i=1}^N{q_i}=1,\  0<q_i\le 1,
\notag\\
&~\mathrm{C2:}~\mathbb{E}  [T^{(r)}(\bm{q},L_{c}^{ i }) ] \leq T,\notag\\
&~\mathrm{C3:}~
L_{c,\min}\leq L_c^i\leq L,\ i\in[1,N].
\end{align}
Constraint  $\mathrm{C1}$ ensures that $q_i$
  represents a valid discrete probability distribution over $N$  clients and guarantees the unbiased global model. %
  Moreover, Constraint $\mathrm{C2}$ ensures that the per-round latency is within a pre-determined threshold and
  Constraint C3 ensures that the split point should be in a feasible range, where $L_{c,\min}$ is the minimum cut layer index that protects user privacy by ensuring sensitive information on the client side.
Besides, $R$ represents the total number of training rounds.
 The expectation  $\mathbb{E}[{\nabla _{\bf{w}}}f({{\bf{w}}^{t - 1}}(\boldsymbol{q}))]$ in Eq.~\eqref{ob1} arises from the randomness in client sampling probability  $\boldsymbol{q}$ and the local SGD. %

 \section{Solution Approach}\label{solu_appro}

In this section, we devise an efficient iterative method and derive the optimal solution to Problem \textbf{P1}. 

\subsection{Analysis of the Expected Round Time}
We first establish the following theorem for expected round time.
\begin{theorem}
The expected per-round training latency with the sampled clients is
\begin{align}
&\mathbb{E} [T^{(r)}(\boldsymbol{q},L_c^i)]= \notag \\=&K
\mathbb{E} \left[\frac{1}{K} \sum_{i\in \mathcal{K} ^{(r)}}\underset{A_i( L_{c}^{i} )}{\underbrace{(t_{i}^{u}( L_{c}^{i}) +t_{i}^{d}( L_{c}^{i} ) +\tau _{i}^{c}( L_{c}^{i}) +\tau _{i}^{s}( L_{c}^{i} ) )}} \right] \notag \notag\\
= &K\sum_{i=1}^N{q_iA_i( L_{c}^{i} )},
\end{align}
where $t_{i}^{u}$ denotes the per-round latency of transmitting activations from client $i$ to edge server, and $t_{i}^{d}$ denotes per-round latency of transmitting activations' gradients from edge server to client $i$, $\tau _{i}^{c}$ denotes the client-side per-round computing latency, and $\tau _{i}^{s}$ denotes the server-side per-round computing latency for client $i$.
\end{theorem}
\begin{proof}
    See \cite{10443546}.
\end{proof}

This theorem is crucial for our subsequent design because the expected latency will be incorporated into the optimization framework. In particular, the overall expected training latency appears as Constraint C2, $\mathbb{E}\big[T^{(r)}(\boldsymbol{q},L_c^i)\big] \le T$, in our original formulation of Problem \textbf{P1}. By substituting the above expression, the latency constraint becomes $K\sum_{i=1}^N q_i \, A_i(L_c^i) \le T$.

\subsection{Solution Overview}
\textit{b1) Find Optimal $\boldsymbol{q}_{L_c}$ and $L_c^{i}$ given $L_c$:} To solve Problem \textbf{P1}, we first exploit the discrete nature of the maximum cut layer index \(L_c\in\{ L_{c,\min}, L_{c,\min}+1, \ldots, L-1, L \}\).  By simply enumerating all \(L-L_{c,\min}+1\) layers, we guarantee a global optimum over \(L_c\) with complexity of \(O(L)\).  
Let $q_{i}(L_c)$ denote the value of $q_i$ when the maximum cut layer index is $L_c$.
For each trial value \(L_c\), we  aim to minimize the upper bound in Eq.~\eqref{convergence_bound}:
\begin{equation}\begin{array}{cl}
\label{client-sampling}
\!\!\!\!\!\!\!\textbf{P2}:\quad &\underset{\{ \boldsymbol{q}_{L_c}, L_{c}^{ i }
\}} \min   
\mathcal{U}_{L_{c}}(\boldsymbol{q}_{L_c}) \\
\quad\quad \text { s.t. } 
&~\mathrm{C1}^\prime:~
\sum\limits_{i=1}^N{q_{i}(L_c)}=1,\  0<q_{i}(L_c)\le 1,
 \\
&~\mathrm{C2}^\prime:~ 
K\sum\limits_{i=1}^N{q_{i}(L_c)A_{i}(L_c^i)}
\leq T,\\
&~\mathrm{C3}^\prime:~\max\limits_i\{L_c^i\}=L_c,\
L_{c,\min}\leq L_c^i\leq L_c,\\ 
&~\quad \quad \quad \quad \quad   \quad \quad \quad \quad   \quad \quad \quad   \quad \quad \quad   i \in[1,N],
\end{array}\end{equation}
where 
\(
\mathcal{U}_{L_c}(\boldsymbol q_{L_c})
\triangleq
\mathcal U\!\left(\boldsymbol q_{L_c},\{L_c^i\}_{i=1}^N\right),
\) and
\(\max_i L_c^i=L_c\).

To tackle Problem\textbf{ P2}, we first introduce the auxiliary variable $M_{L_c}$ satisfying $M_{L_c}\geq
\underset{i^\prime}{\max}\bigg\{ \frac{{m_{i^\prime}}^2}{q_{i,L_c^i}\big( 1-p_{i^\prime} \big)\big( 1-\varphi_{i^\prime} \big)\big( 1-a_{i^\prime} \big)}\bigg\}
$ ($\forall i^\prime \in \mathcal{N}$) and transform Problem\textbf{ P2} into
\begin{equation}\begin{array}{cl}
\label{client-sampling_0}
\textbf{P3}:\quad &\underset{\{ \boldsymbol{q}_{L_c}, L_{c}^{ i },M_{L_c}
\}} \min  \mathcal{U}_{L_{c}}(\boldsymbol{q}_{L_c},M_{L_c}) 
 \\
 \text { s.t. } 
&\mathrm{C1}^\prime\text{-}\mathrm{C3}^\prime,\\
 &\mathrm{C4:}~
\dfrac{{m_{i}}^2}{q_{i,L_c}\big( 1-a_{i} \big)\big( 1-p_{i} \big) \big( 1-\varphi_{i} \big)}\leq  M_{L_c},\ \forall i \in \mathcal{N}
,
\end{array}\end{equation}
where the objective can be rewritten as
\begin{align}
    \mathcal{U}_{L_{c}}(\boldsymbol{q}_{L_c},M_{L_c})=\sum_{i=1}^N \frac{m_i^2}{q_{i}(L_c)}\,\widetilde{C}_{i,L_c}(M_{L_c})+
\frac{2\vartheta}{\gamma R}
,
\end{align}
with the coefficient \(\widetilde{C}_{i,L_c}(M_{L_c})\) 
being
\begin{align}
&\widetilde{C}_{i,L_c}(M_{L_c})= \frac{\beta \gamma}{\big( 1-p_i \big)}\bigg\{ \frac{\sum_{j=1}^{L_c}{(}\sigma _{j}^{2}+G_{j}^{2})}{\big( 1-\varphi _i \big)\big( 1-a_i \big)}+\notag\\&\sum_{j=L_c+1}^L{(}\sigma _{j}^{2}+G_{j}^{2}) \bigg\} 
+\beta ^22\gamma ^2I^2\bigg( NM_{L_c}+\frac{1}{\big( 1-a_i \big)}\cdot\notag\\
&\frac{1}{\big( 1-p_i \big) \big( 1-\varphi _i \big)} \bigg) \sum_{j=1}^{L_c}{G_{j}^{2}}.
\end{align}

Since \(L_c\) is fixed, we observe that the optimized varible $L_c^i$  appears only in Constraint \(\text{C2}\). 
We assume that, for each client $i$, the cut layer $L_c^i$ is selected to minimize its associated cost $A_i$. Let $A_i^*$ denote this minimum, defined as
$A_i^* \triangleq \min_{L_c^i\in[L_{c,\min},L_c]} A_i(L_c^i)$. Thus, Problem $\textbf{P3}$ can be reformulated as:
\begin{equation}\begin{array}{cl}
\textbf{P4}:\quad &\underset{\{ \boldsymbol{q}_{L_c}, M_{L_c}
\}} \min  \mathcal{U}_{L_{c}}\big(\boldsymbol{q}_{L_c},M_{L_c}|\{A_i^*\}_{i\in\mathcal{N}}\big) 
 \\
 \text { s.t. } 
&\mathrm{C1}^\prime:~
\sum\limits_{i=1}^N{q_{i}(L_c)}=1,\  0<q_{i}(L_c)\le 1,
 \\
&\mathrm{C2}^{\prime \prime}:~K\sum\limits_{i\in{\mathcal{P}_c}}{q_{i}(L_c)\cdot A_{i}^*}
\leq T,\\
 &\mathrm{C4:}~
\dfrac{{m_{i}}^2}{q_{i}(L_c)\cdot\big( 1-a_{i} \big)\big( 1-p_{i} \big) \big( 1-\varphi_{i} \big)}\leq  M_{L_c},\\
 & \quad \quad \quad \quad \quad \quad \quad \quad \quad \quad \quad \quad \quad \quad \quad \quad \ \ \forall i \in \mathcal{N}
.
\end{array}\end{equation}

\textit{b2) Divide Client Set $\mathcal{N}$ into Positive Subset ${\mathcal{P}_c}$ and Negative Subset $\mathcal{N}_c$:} Based on the sign of \(\widetilde C_{i,L_c}(M_{L_c})\), we define
$
{\mathcal{P}_c} = \{ i \mid \widetilde C_{i,L_c}(M_{L_c}) > 0 \} 
$
and
$
\mathcal{N}_c = \{ i \mid \widetilde C_{i,L_c}(M_{L_c}) < 0 \} 
$.
Problem \textbf{P4} can be decomposed according to the subsets ${{\mathcal{P}_c}}$ and $\mathcal{N}_c$, which is further expressed by
\begin{equation}\begin{array}{cl}
\label{client-sampling_1}
\!\textbf{P4}^\prime: &\underset{\{ \boldsymbol{q}_{L_c}^{\mathcal{P}_c},M_{L_c}
\}} \min  \mathcal{U}^{\mathcal{P}_c}_{L_{c}}\big(\boldsymbol{q}_{L_c}^{\mathcal{P}_c},M_{L_c}|\{A_i^*\}_{i\in\mathcal{N}}\big) 
\\
 \text { s.t. } 
&~\mathrm{C1}^{\prime \prime}:~
\sum\limits_{i\in{\mathcal{P}_c}}{q_{i}(L_c)}=1-\sum\limits_{j\in\mathcal{N}_c} q_{j,L_c},\\ &\quad \quad \quad \quad  \quad \quad \quad \quad \quad \quad \quad \quad 0<q_{i}(L_c)\le 1,\\
&~\mathrm{C2}^{\prime \prime}:~K\sum\limits_{i\in{\mathcal{P}_c}}{q_{i}(L_c)\cdot A_{i}^*}
\leq T,\\
 &~\mathrm{C4}^{\prime \prime}:~
\dfrac{{m_{i}}^2}{q_{i}(L_c)\cdot \left( 1-a_{i} \right)\left( 1-p_{i} \right) \left( 1-\varphi_{i} \right)}\leq  M_{L_c},\\& \quad \quad \quad \quad \quad \quad \quad \quad \quad \quad \quad \quad \quad \quad \quad \quad   \forall i \in {\mathcal{P}_c} \cup \mathcal{N}_c
,
\end{array}\end{equation}
and
\begin{equation}\begin{array}{cl}
\label{client-sampling_2}
\!\textbf{P4}^{\prime\prime}: &\underset{\{ \boldsymbol{q}_{L_c}^{\mathcal{N}_c},M_{L_c}
\}} \min  \quad 
\mathcal{U}^{\mathcal{N}_c}_{L_{c}}\big(\boldsymbol{q}_{L_c}^{\mathcal{N}_c},M_{L_c}|\{A_i^*\}_{i\in\mathcal{N}}\big) \\
 \text { s.t. } 
&~\mathrm{C1}^{\prime \prime \prime}:~
\sum\limits_{i\in\mathcal{N}_c}{q_{i}(L_c)}=1-\sum\limits_{j\in{\mathcal{P}_c}} q_{j,L_c},\\ &\quad \quad \quad \quad  \quad \quad \quad \quad \quad \quad \quad \quad 0<q_{i}(L_c)\le 1,\\
&~\mathrm{C2}^{\prime  \prime\prime}:~K\sum\limits_{i\in \mathcal{N}_c}{q_{i}(L_c)\cdot A_{i}^*}
\leq T,\\
 &~\mathrm{C4}^{\prime \prime \prime}:~
\dfrac{{m_{i}}^2}{q_{i}(L_c)\cdot \left( 1-a_{i} \right)\left( 1-p_{i} \right) \left( 1-\varphi_{i} \right)}\leq  M_{L_c},\\& \quad \quad \quad \quad \quad \quad \quad \quad \quad \quad \quad \quad \quad \quad \quad \quad   \forall i \in {\mathcal{P}_c} \cup \mathcal{N}_c
,
\end{array}\end{equation}
where $
\boldsymbol{q}_{L_c}^{{\mathcal{P}_c}}
\;\triangleq\;
\bigl(q_{i}(L_c)\bigr)_{\,i\in{\mathcal{P}_c}}
$ and $
\boldsymbol{q}_{L_c}^{\mathcal{N}_c}
\;\triangleq\;
\bigl(q_{i}(L_c)\bigr)_{\,i\in\mathcal{N}_c}
$. 
Detailed solutions to Problem $\textbf{P3}^\prime$ and Problem $\textbf{P3}^{\prime\prime}$ are provided in Section~\ref{solu_appro}.C and Section~\ref{solu_appro}.D, respectively.
Moreover,
we have the following theorem.
\begin{theorem}
Solving Problem $\mathbf{P4}$ is equivalent to  solving Problem $\mathbf{P4}^{\prime}$ and Problem $\mathbf{P4}^{\prime \prime}$, respectively.
\end{theorem}

\begin{proof}
For the positive subset \({\mathcal{P}_c}\), each term in \(\mathcal{U}_{L_c}^{{\mathcal{P}_c}}\) is strictly convex with respect to \(q_i(L_c)\). Thus, standard convex optimization techniques guarantee a unique minimizer for these components.
For the subset \(\mathcal{N}_c\), each term in \(\mathcal{U}_{L_c}^{\mathcal{N}_c}\) is monotonically increasing in \(q_i(L_c)\). Thus, the optimal solution is achieved by setting \(q_i(L_c)\) to its lower bound.
Additionally, the global normalization constraint
\[
\sum_{i\in\mathcal{N}}q_i(L_c)=\sum_{i\in{\mathcal{P}_c}}q_i(L_c)+\sum_{j\in\mathcal{N}_c}q_j(L_c)=1
\]
as well as any delay constraints are linear. This indicates that the respective solutions to \(\mathbf{P4}'\) and \(\mathbf{P4}''\) satisfy the overall normalization and delay constraints.
Crucially,  the auxiliary variable \(M_{L_c}\) couples the subproblems together, ensuring that the feasible solutions from both \({\mathcal{P}_c}\) and \(\mathcal{N}_c\) spaces can be integrated into a valid global solution.
Thus, by solving
\(
\mathbf{P4}' \) and \(\mathbf{P4}''
\), the resulting solution is globally optimal for Problem \(
\mathbf{P4}
\).
\end{proof}

\textit{b3) Ensure one client-side model is split at \(L_c\):}  To satisfy Constraint $\mathrm{C3}^\prime$ in  Problem \textbf{P3}, we need to force exactly that one client adopts \(L_c\) with minimal cost. Specifically, we define the following problem
 \begin{align}
\textbf{P5}: &\underset{\{ \boldsymbol{q}_{L_c}, M_{L_c}
\}} \min  \mathcal{U}_{L_c}\big(\boldsymbol{q}_{L_c},M_{L_c}|\{\{A_{i^\prime}^*\}_{{i^\prime}\in\mathcal{N},{i^\prime}\neq { i}},A_{ i}(L_c)\}\big)
 \notag\\
 \text { s.t. } 
&\mathrm{C1}^\prime,\ \mathrm{C4},\notag\\
&\mathrm{C2}^{\dagger}:~K\big({\sum\limits_{{i^\prime}\in\mathcal{N},{i^\prime}\neq{ i}}{q_{{i^\prime}}(L_c)A_{{i^\prime}}^*}}+q_{{ i},L_c}A_{{ i}}(L_c)\big)
  \leq T.
\end{align}


\begin{theorem}
Solving Problem $\mathbf{P3}$ is equivalent to first solving Problem $\mathbf{P4}^{\prime}$ and Problem $\mathbf{P4}^{\prime \prime}$, and then setting $ L_{c}^{i_0} =  L_c\,$, where 
\begin{align}
  &i_0 \;=\;\arg\min_{{ i}}\ \mathcal{U}_{L_c}\big(\boldsymbol{q}_{L_c}^*,M_{L_c}^*|\big\{\{A_{i^\prime}^*\}_{{i^\prime}\in\mathcal{N},{i^\prime}\neq { i}},A_{ i}(L_c)\big\}\big), \notag \\
  &\quad \quad \quad \quad \quad \quad  \quad \quad \quad \quad \quad \quad \quad \quad \quad \quad \quad \quad \quad \quad  \quad \forall { i} \in \mathcal{N}
,
\label{i_0}
\end{align}
and where $\mathcal{U}_{L_c}\big(\boldsymbol{q}_{L_c}^*,M_{L_c}^*|\big\{\{A_{i^\prime}^*\}_{{i^\prime}\in\mathcal{N},{i^\prime}\neq { i}},A_{ i}(L_c)\big\}\big)$ represents the optimal solution to Problem $\mathbf{P5}$.
This guarantees the constraint $\mathrm{C3}^\prime$, i.e., \(\max\limits_i L_c^i = L_c\), while minimizing the overall upper bound of global training loss.
\end{theorem}
\begin{proof}
 By transforming $\mathrm{C3}^\prime$, the decomposed subproblems $\textbf{P4}^\prime$ and $\textbf{P4}^{\prime \prime}$ yield a globally optimal solution $\{A_i^*\}_{i\in\mathcal{N}}$ under the relaxed constraints. If some clients satisfy $L_c^i = L_c$, $\mathrm{C3}^\prime$, inherently holds; otherwise, constraint recovery is required.
Since we force client $i$ to adopt $L_c$, which increases its latency $A_{{ i}}(L_c) \geq A_{{ i}}^*  $, potentially violating the latency constraint $\text{C2}^\prime$, we re-optimizing  Problem  \textbf{P5} by imposing Constraint $\mathrm{C2}^{\dagger}$.
Specifically, we move each client \(i\)'s splitting point from its current value \(L_c^i\) to \(L_c\) (with all other clients' splitting points remaining unchanged), and then compute the minimum value of the new optimization Problem  \textbf{P5}.
Since there are \(N\) clients, this computation is repeated \(N\) times.
 By comparing \(N\) corresponding minimum values across different clients with  Eq.~\eqref{i_0}, we can identify the client \(i_0\) that results in the smallest increase. 
 Thus, 
the scheme guarantees the constraint \(
\max_{i}\{L_c^i\} = L_c\ ( L_{c,\min} \leq L_c^i \leq L_c)
\) and preserves the global optimality of Problem \(\mathbf{P3}\).

\end{proof}
\textcolor{black}{Fig. \ref{alg_flow_optimal_solution_fig} gives the high-level workflow, while Algorithm \ref{alg:short} provides the step-by-step implementation details. Specifically, the flowchart presents a three-level nested optimization procedure: (i) an outer loop that enumerates candidate maximum cut layers, (ii) a middle loop that performs a bisection search for the auxiliary variable 
$M$, and (iii) an inner loop that optimizes the client sampling probabilities 
$\bm q$ (via updating the Lagrange multipliers) for a fixed 
$M$. In the inner loop, clients are first classified into positive and negative subsets, and the corresponding optimal sampling probabilities are computed, respectively. 
The algorithm further computes the normalization/latency errors, updates 
$\lambda$ and 
$\nu$, and repeats until the error tolerances are satisfied. 
After the bisection on 
$M$ converges, it checks the feasibility condition \(\max\limits_{i} \{L_{c,i}\} < L_c\). If the condition holds, \textbf{Theorem 4} is applied to enforce feasibility. Finally, the algorithm updates the best solution if improved, otherwise it proceeds to the next candidate in the outer loop and terminates when all candidates have been examined.}
\begin{figure}[t]
\centering\includegraphics[width=8.5cm]{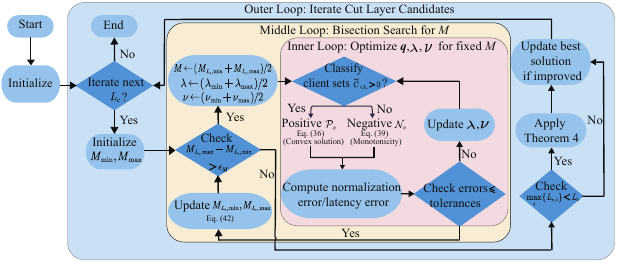}
        
	\centering\caption{ \textcolor{black}{Flowchart of the optimal client sampling and model splitting.} }
		\label{alg_flow_optimal_solution_fig}
        \vspace{-0.8em}
\end{figure}

\subsection{Solution to Problem \texorpdfstring{$\mathbf{P4}^{\prime}$}{P4'} of Positive Subset \texorpdfstring{${\mathcal{P}_c}$}{P\_c}}
In the following, we provide the solution approach corresponding to positive subset $\mathcal{P}_c$.

\textit{c1) Inner Algorithm - Convex Subproblem Solver:}
For each \(i \in {\mathcal{P}_c}\), we define
\(
f_i^{\mathcal{P}_c}(q_{i}(L_c)) \;=\; \frac{m_i^2\,\widetilde{C}_{i,L_c}(M_{L_c})}{q_{i}(L_c)}
\).
Since \(q_{i}(L_c)>0\), \(m_i^2>0\), and \(\widetilde{C}_{i,L_c}(M_{L_c})>0\), it follows that
\( f_i^{{\mathcal{P}_c}\prime\prime}(q_{i}(L_c)) > 0 \) for all \( q_{i}(L_c) > 0 \).
This strictly positive second-order derivative implies that each \( f_i^{\mathcal{P}_c}(q_{i}(L_c)) \) is strictly convex over its domain.
Thus, we introduce Lagrange multipliers:
\begin{itemize}
    \item \(\lambda\) for the normalization constraint \(\sum\nolimits_{i\in{\mathcal{P}_c}}{q_{i}(L_c)}=1-\sum\nolimits_{j\in\mathcal{N}_c} q_{j,L_c}\) ($\text{C1}^{\prime\prime}$),
    \item \(\nu\) for the delay constraint \(K\sum_{i=1}^N{q_{i}(L_c)A_i^*}
\leq T\) ($\text{C2}^{\prime\prime}$).
\end{itemize}

Assume that Constraint $\text{C4}^{\prime\prime}$ is achieved by our choice of \(M_{L_c}\), the Lagrangian function can be expressed as
\begin{align}
\mathcal{L}(\bm{q}_{L_c}^{\mathcal{P}_c},\lambda,\nu) = &\sum_{i\in{\mathcal{P}_c}} \frac{m_i^2}{q_{i}(L_c)}\,\widetilde{C}_{i,L_c}(M_{L_c}) + \lambda\bigg(\sum_{i\in{\mathcal{P}_c}} q_{i}(L_c)+\notag\\&\zeta - 1\bigg) + \nu\left(K\sum_{i=1}^N{q_{i}(L_c)A_i^*}
-T\right) + \frac{2\vartheta}{\gamma R},
\end{align}
where $\zeta=\sum_{j\in\mathcal{N}_c} q_{j,L_c}$.

Taking the first-oder derivative with respect to \(q_{i}(L_c)\) (for \(q_{i}(L_c)>0,\ i\in{\mathcal{P}_c}\)) and setting it to zero, we derive
\begin{align}
-\frac{m_i^2\,\widetilde{C}_{i,L_c}(M_{L_c})}{q_{i}(L_c)^2} + \lambda + \nu\,(KA_i^*) = 0.
\label{q_i_inner}
\end{align}
Moreover, Constraint $\text{C4}^{\prime\prime}$ requires that
\begin{align}
q_{i}^*(L_c) \geq \frac{m_i^2}{M_{L_c}(1-a_i)(1-p_i)(1-\varphi_i)},\quad \ \forall i \in \mathcal{N}.
\label{q_i_c4}
\end{align}

Thus, combining \eqref{q_i_inner} and \eqref{q_i_c4}, the semi-closed form solution for clients in \({\mathcal{P}_c}\) is given by
\begin{align}
q_{i}^*(L_c,M_{L_c}) = &\max\Biggl\{ \frac{m_i^2}{M_{L_c}(1-a_i)(1-p_i)(1-\varphi_i)}, \notag
\\
&\quad \quad  \quad     \quad\sqrt{\frac{m_i^2\,\widetilde{C}_{i,L_c}(M_{L_c})}{\lambda+\nu\,(KA_i^*)}} \Biggr\},\ \forall i\in{\mathcal{P}_c}.
\label{q_closed}
\end{align}

\textit{c2) Outer Algorithm - Nested Bisection Search:}
In the outer loop, we initialize lower and upper bounds \(M_{L_c,\min}\) and \(M_{L_c,\max}\), respectively, based on theoretical limits and set \(M_{L_c} = (M_{L_c,\min}+M_{L_c,\max})/{2}\). For the fixed \(M_{L_c}\), the inner loop adjusts the Lagrange multipliers \(\lambda\) and \(\nu\) via bisection \cite{eiger1984bisection} to ensure that the computed \(\bm q_{L_c}^{\mathcal{P}_c}\) satisfies both the normalization constraint \(\sum_{i=1}^N q_{i}(L_c) = 1\) and the latency constraint \(K\sum_{i=1}^N{q_{i}(L_c)A_i^*}
<T\). Specifically, for each client \(i\) with fixed $M_{L_c}$, \(q_{i}(L_c,M_{L_c})\) is updated using the closed-form expression in \eqref{q_closed}.
After obtaining \(\bm{q}_{L_c}^{\mathcal{P}_c}(M_{L_c})\), we provide the candidate $M_{L_c,\text{can}}^{\mathcal{P}_c}$ for updating $M_{L_c}$  as follows:
\begin{equation}
M_{L_c,\text{can}}^{\mathcal{P}_c} = \max_{i\in{\mathcal{P}_c}} \left\{ \frac{m_i^2}{q_{i}(L_c,M_{L_c})(1-a_i)(1-p_i)(1-\varphi_i)} \right\}.
\label{eq:M_update_positive}
\end{equation}
Furthermore, \(M_{L_c,\text{can}} = \max\{M_{L_c,\text{can}}^{\mathcal{P}_c},M_{L_c,\text{can}}^{\mathcal{N}_c}\}\) and $M_{L_c,\text{can}}^{\mathcal{N}_c}$ is obtained in Section~\ref{solu_appro}.D.
If \(M_{L_c,\text{can}} \le M_{L_c}\), this indicates that the current \(M_{L_c}\) is overly conservative (i.e., larger than necessary to satisfy the constraints), so we reduce the upper bound by setting \(M_{\max} = M_{L_c}\); otherwise, if \(M_{L_c,\text{can}} > M_{L_c}\), the current \(M_{L_c}\) is too small, and we update the lower bound by setting \(M_{\min} = M_{L_c}\). The nested bisection process continues until \(M_{\max} - M_{\min} < \epsilon_M\).
It is obvious that when achieving convergence, the optimal \(M_{L_c}^*\) satisfies
\begin{equation}
M_{L_c}^* = \max_{i\in\mathcal{N}} \left\{ \frac{m_i^2}{q_{i}^*(L_c)\cdot(1-a_i)(1-p_i)(1-\varphi_i)} \right\}.
\label{M_P}
\end{equation}

\begin{algorithm}[t]
\begin{spacing}{0.8}
\caption{Optimal  Model Splitting  and Client Sampling Algorithm.}
\label{alg:short}
	\renewcommand{\algorithmicrequire}{\textbf{Input:}}
	\renewcommand{\algorithmicensure}{\textbf{Output:}}
\begin{algorithmic}[1]
\REQUIRE For each client \(i=1,\dots,N\): parameters \(m_i,\,a_i,\,p_i,\,\varphi_i,\,\widetilde{C}_i,\); system constants \(K,\,R,\,T\); tolerances \(\epsilon_M,\,\epsilon_\lambda,\,\epsilon_\nu\); initial bounds \(M_{\min}\) and \(M_{\max}\) for \(M\); bisection intervals \([\lambda_{\min},\lambda_{\max}]\)  and \\ \([\nu_{\min},\nu_{\max}]\); candidate set of maximum split points \\ \(\mathcal{L}_c\); and \(L_{c,i}^{\min}\) for each client.
\ENSURE Optimal maximum split point \(L_c^*\), auxiliary \\ variable \(M^*\), decision vector \(\bm{q}^*\), and cut \\ layers \(\{L_{c,i}^*\}\).
\STATE \(F_{\text{best}} \gets \infty\).
\FOR{each candidate \(L_c \in \mathcal{L}_c\)}
  \STATE \textbf{// Step 1: Optimize \(M_{L_c}\) and compute \(\bm{q}_{L_c}\) via nested bisection}
  \STATE Set \(M_{L_c,\min} \gets M_{\min}\) and \(M_{L_c,\max} \gets M_{\max}\).
  \WHILE{\(M_{L_c,\max} - M_{L_c,\min} > \epsilon_M\)}
    \STATE Initialize \(M \gets {(M_{L_c,\min} + M_{L_c,\max})}/{2}\), \(\lambda \gets {(\lambda_{\min}+\lambda_{\max})}/{2}\), \(\nu \gets ({\nu_{\min}+\nu_{\max}})/2\).
    \REPEAT
      \FOR{each client \(i=1,\dots,N\)}
          \STATE Set \(q_{i}(L_c)\) using Eq.~\eqref{q_closed} when \(\widetilde{C}_{i,L_c}>0\).
          \STATE Set \(q_{i}(L_c)\) using Eq.~\eqref{q_another} when \(\widetilde{C}_{i,L_c}<0\).
      \ENDFOR
      \STATE Compute normalization error: \(e_1 \gets \sum\limits_{i=1}^N q_i - 1\).
      \STATE Compute latency error: \(e_2 \gets K\sum\limits_{i=1}^N q_i\, A_i^* - T\).
      \STATE Adjust \(\lambda\) and \(\nu\) according to \(e_1\) and \(e_2\).
    \UNTIL{\(|e_1|\le \epsilon_\lambda\) and \(|e_2|\le \epsilon_\nu\)}
    \STATE Update \(M_{L_c,\min}\) or \(M_{L_c,\max}\) using the candidate update rule Eq. \eqref{m_update};
  \ENDWHILE
  
  \STATE \textbf{// Step 2: Determine each client's cut layer}
  \IF{\(\max\limits_{i} \{L_{c,i}\} < L_c\)}
    \STATE Identify \(i_0 \) according to \textbf{Theorem 4} and force \(L_{c}^{i_0} \gets L_c\).
  \ENDIF
  
  \STATE  Update best solution if current candidate improves \\ the objective: set \(L_c^* \gets L_c\), \(M^* \gets M_{L_c}\), \\ \(\bm{q}^* \gets \bm{q}_{L_c}\), and \(\{L_{c,i}^*\} \gets \{L_{c,i}\}\).
\ENDFOR
\RETURN \(L_c^*,\, M^*,\, \bm{q}^*,\, \{L_{c,i}^*\}\).
\end{algorithmic}
\end{spacing}
\end{algorithm}

\subsection{Solution to Problem \texorpdfstring{$\mathbf{P4}^{\prime\prime}$}{P4''} of Negative Subset \texorpdfstring{$\mathcal{N}_c$}{N\_c}}
Next, we provide the solution for the case of negative subset $ \mathcal{N}_c$.
For each \(i \in \mathcal{N}_c\), we define
\(
f_i^{\mathcal{N}_c}(q_{i}(L_c)) \;=\; \frac{m_i^2\,\widetilde{C}_{i,L_c}(M_{L_c})}{q_{i}(L_c)}
\).
When fixed $M_{L_c}$, the second-order derivative of some client $i\ (i \in {\mathcal{N}_c})$ is
$
 f_i^{\mathcal{N}_c\prime\prime}(q_{i}(L_c))
 = \frac{2\,m_i^2\,\widetilde{C}_{i,L_c}}{q_i^3} < 0$,
    which implies that $f_i^{\mathcal{N}_c}(q_{i}(L_c))$ is concave.
    The first-order derivative is $ f_i^{\mathcal{N}_c\prime}(q_{i}(L_c))= -\frac{m_i^2\,\widetilde{C}_i}{q_i^2}  > 0
    $, and hence $ f_i^{\mathcal{N}_c}(q_{i}(L_c))$ is monotonically increasing with respect to $q_{i}(L_c)$.
Thus, the minimization of objective function $\mathcal{U}^{\mathcal{N}_c}_{L_{c}}\big(\boldsymbol{q}_{L_c}^{\mathcal{N}_c},M_{L_c}|\{A_i^*\}_{i\in\mathcal{N}}\big)$ is achieved by choosing the smallest feasible \(q_{i}(L_c)\), which is lower bounded by Constraint $\mathrm{C4}^{\prime \prime \prime}$ and is expressed by
\begin{equation}
q_{i}^*(L_c,M_{L_c}) = \frac{m_i^2}{M_{L_c}}\frac{1}{(1-a_i)(1-p_i)(1-\varphi_i)},\quad \forall\, i \in \mathcal{N}_c.
\label{q_another}
\end{equation}
Thus, for clients in \(\mathcal{N}_c\), \(M_{L_c,\text{can}}^{\mathcal{N}_c}\) has the closed-form solution, which is given by
\begin{equation}
M_{L_c,\text{can}}^{\mathcal{N}_c} = \sum_{i\in\mathcal{N}_c} \frac{m_i^2}{(1-a_i)(1-p_i)(1-\varphi_i)}.
\label{eq:M_closed_negative}
\end{equation}




\begin{theorem}
Algorithm 2, with the  solution \( \bm{q}_{L_c}^*\) in each iteration   given by 
\begin{align}
q_{i}^*(L_c) = \begin{cases}
\displaystyle \max\!\Biggl\{ \frac{m_i^2}{M_{L_c}}\frac{1}{(1-a_i)(1-p_i)(1-\varphi_i)},\\ \quad \quad \quad \quad \ \sqrt{\frac{m_i^2\,\widetilde{C}_{i,L_c}(M_{L_c})}{\lambda + \nu\,(KA_i^*)}} \Biggr\},\  \forall i\in{\mathcal{P}_c}, \\
\displaystyle \frac{m_i^2}{M_{L_c}}\frac{1}{(1-a_i)(1-p_i)(1-\varphi_i)}, \  \forall i\in\mathcal{N}_c,
\end{cases}
\end{align}
and the updating rule for $M_{L_c}$ given by
\begin{align}
M_{L_c,\mathrm{can}} = \max \bigg\{&\max_{i\in{\mathcal{P}_c}} \big\{ \frac{m_i^2}{q_{i}^*(L_c)}\frac{1}{(1-a_i)(1-p_i)(1-\varphi_i)} \big\}, \notag\\
& \quad \sum_{i\in\mathcal{N}_c} \frac{m_i^2}{(1-a_i)(1-p_i)(1-\varphi_i)}
\bigg\},
\label{m_update}
\end{align}
obtains the global optimal solution to Problem $\mathbf{P4}$.
\end{theorem}

\begin{proof}
For clients in the subset \({\mathcal{P}_c}\), the objective function term is strictly convex in \(q_{i}(L_c)>0\). In addition, the normalization constraint 
\(
\sum_{i\in{\mathcal{P}_c}} q_{i}(L_c) = 1-\sum_{j\in\mathcal{N}_c} q_{j,L_c}
\)
and the delay constraint 
\(
K\sum_{i\in{\mathcal{P}_c}} q_{i}(L_c)A_i^* \le T
\)
are linear. Hence, for any fixed \(M_{L_c}\), the inner problem (i.e., Problem \(\mathbf{P4}'\)) forms a convex optimization problem, ensuring that, by the KKT conditions, a unique optimal solution \(\bm{q}_{L_c}^{{\mathcal{P}_c}*}(M_{L_c})\) is obtained from Eq.~\eqref{q_closed}.
For clients in the negative subset \(\mathcal{N}_c\), the objective function is monotonically increasing in \(q_{i}(L_c)\). Thus, the optimal solution \(\bm{q}_{L_c}^{\mathcal{N}_c*}(M_{L_c})\) is obtained from Eq.~\eqref{q_another}.
Notably, even if one considers increasing \(q_{i}(L_c)\) for clients in \(\mathcal{N}_c\) to offer additional slack for devices in \({\mathcal{P}_c}\) (in order to satisfy the normalization constraint \(\sum_{i=1}^{N} q_{i}(L_c)=1\)), such an adjustment would lead to an overall increase in the objective function in two ways: on the one hand, the objective contributions of the devices in \(\mathcal{N}_c\) would directly increase;
on the other hand, \(q_{i}(L_c)\) (\(i \in {\mathcal{P}_c}\)) has to decrease. Since the objective function for the positive subset is strictly convex and decreasing with respect to  \(q_{i}(L_c)\), this would further increase the overall objective value.
Thus, we conclude that selecting the lower bound for \(q_{i}(L_c)\) in the negative subset is the optimal strategy. As a result,  \(\bm{q}_{L_c}^*\) satisfy all the constraints and minimize the original objective function. Moreover, \(M_{L_c}\) is iteratively updated using a nested bisection method (see Eq.~\eqref{m_update}) until the absolute difference between the candidate update value \(M_{L_c,\mathrm{can}}\) and the current \(M_{L_c}\) is less than the predefined tolerance \(\epsilon_M\); i.e., until \(|M_{L_c,\mathrm{can}} - M_{L_c}| < \epsilon_M\), indicating convergence. 
In summary, the combination of decomposing the problem (with convex optimization for \({\mathcal{P}_c}\) and the lower bound selection for \(\mathcal{N}_c\)) and the nested bisection update for \(M_{L_c}\) ensures that the overall solution \((M^*_{L_c},\bm{q}_{L_c}^*)\) satisfies all constraints while minimizing the objective function, thereby achieving global optimality.

\end{proof}

\begin{theorem}
The overall computational complexity of \textbf{Algorithm \ref{alg:short}} is
\[
O\Bigl( N L^2 + L \cdot N\,\log\frac{1}{\epsilon_M}\,\log\frac{1}{\epsilon_\lambda}\,\log\frac{1}{\epsilon_\nu}\Bigr).
\]
\end{theorem}

\begin{proof}
The overall complexity is analyzed as follows.
Enumerating the maximum split point \(L_c \in \{ L_{c,\min}, L_{c,\min}+1, \ldots, L-1, L \}\) requires \(O(L)\) iterations. For each fixed \(L_c\), each of \(N\) clients must enumerate all candidate local split layers in the range \(\{ L_{c,\min}, L_{c,\min}+1, \ldots, L-1, L \}\) to determine its optimal \(L_c^i\),  result in \(O(NL)\). Furthermore, for each fixed \(L_c\), an outer bisection search is performed on \(M_{L_c}\) that converges in \(O\left(\log\frac{1}{\epsilon_M}\right)\) iterations~\cite{Boyd_Vandenberghe_2004}. Moreover, within each outer iteration the Lagrange multipliers \(\lambda\) and \(\nu\) are adjusted via inner bisection searches converging in \(O\left(\log\frac{1}{\epsilon_\lambda}\right)\) and \(O\left(\log\frac{1}{\epsilon_\nu}\right)\) iterations, respectively, with each inner iteration incurring \(O(N)\) time to update the decision variables \(\{q_{i}(L_c)\}\)~\cite{boyd2007notes}. In conclusion, the overall complexity is \(O\bigl(L\cdot [NL + N\,\log\frac{1}{\epsilon_M}\,\log\frac{1}{\epsilon_\lambda}\,\log\frac{1}{\epsilon_\nu}]\bigr) = O\Bigl( N L^2 + L \cdot N\,\log\frac{1}{\epsilon_M}\,\log\frac{1}{\epsilon_\lambda}\,\log\frac{1}{\epsilon_\nu}\Bigr)\), demonstrating that the proposed algorithm achieves an efficient polynomial-time solution.
\end{proof}

\textbf{Insight 4:}  By analyzing the optimal client sampling strategy \(q_{i}^*(L_c)\), we observe that, when \(\widetilde{C}_{i,L_c}>0\), the optimum is determined by taking the maximum between \(\frac{m_i^2}{M}\frac{1}{(1-a_i)(1-p_i)(1-\varphi_i)}\) and an interior solution \(\frac{m_i^2}{M_{L_c}}\frac{1}{(1-a_i)(1-p_i)(1-\varphi_i)}\), while for \(\widetilde{C}_{i,L_c}<0\) the optimal strategy simply adheres to \(\frac{m_i^2}{M_{L_c}}\frac{1}{(1-a_i)(1-p_i)(1-\varphi_i)}\). Moreover, we observe that \( q_i^* \) increases as \( p_i \), \( a_i \), and \( \varphi_i \) increase, which implies that a higher client drop rate results in a higher sampling probability to offset their potential contribution variance. 

\section{Performance Evaluation}\label{simulations}
\subsection{Simulation Setup}

In simulations, we use the datasets EMNIST \cite{cohen2017emnistextensionmnisthandwritten} and CIFAR-10~\cite{lecun1998mnist} to evaluate the performance of the proposed optimal model splitting and client sampling \textbf{(OMS+OCS)} for the model ResNet-50~\cite{he2015deepresiduallearningimage} under both IID and non-IID data distributions. 
Notably, for the non-IID setting, each client is assigned a primary class from which approximately 70\% of its samples are drawn, while the remaining samples are sourced from other classes.
 In particular, we consider a network of \(N=100\) clients, and in each epoch \(K=5\) clients are selected for SFL training. Each client uses a mini-batch size of \(b=32\) with a learning rate \(\gamma=0.0001\). Clients upload their client-side models every \(1\) epoch unless otherwise specified, and the minimum client splitting layer is set to \(L_{c}^{i,\min}=4\).
 In addition, the nested bisection algorithm used to optimize model splitting and client sampling is bounded by the following parameters: $M,\ \lambda,\ \nu=[1 \times 10^{-8},1 \times 10^{7}]$.
The lower bound, \(1 \times 10^{-8}\), is set to avoid numerical underflow and prevent potential division-by-zero issues.
The upper bound, \(1 \times 10^{7}\), ensures that the algorithm can explore sufficiently large space to search all feasible solutions.
 The algorithms have been run on a workstation equipped with an AMD Ryzen Threadripper PRO 5975WX and NVIDIA GeForce RTX 4090. For the reader's convenience, the detailed simulation parameters are summarized in Table \ref{parameter_set}.

\begin{table}[t]\label{table_3}
  \centering
  \caption{Simulation Parameters.}
  \renewcommand{\arraystretch}{0.9}{
  \setlength{\tabcolsep}{0.5mm}{
\begin{tabular}{|c|c|c|c|}
\hline
\textbf{Parameter}          & \textbf{Value} & \textbf{Parameter} & \textbf{Value}  \\ \hline
$N$             & $100  $              & $K$                 & $5 $                         \\ \hline
$b $            & 32              & $\gamma$          & $0.0001$                   \\ \hline
$I$               & $1$           & $L_{c}^{i,\min}$                  & 4                       \\ \hline
$M_{\min}$, $\lambda_{\min}$, $\nu_{\min}$        & $1 \times 10^{-8}$             & $M_{\max}$, $\lambda_{\max}$, $\nu_{\max}$              & $1 \times 10^{7}$                      \\ \hline

\end{tabular}}}
\label{parameter_set}
\vspace{-0.8em}
\end{table}

  \begin{figure*}[t!]
\subfigure[\centering Loss with training round (EMNIST, IID settings).] {
\centering\includegraphics[width=3.9cm]{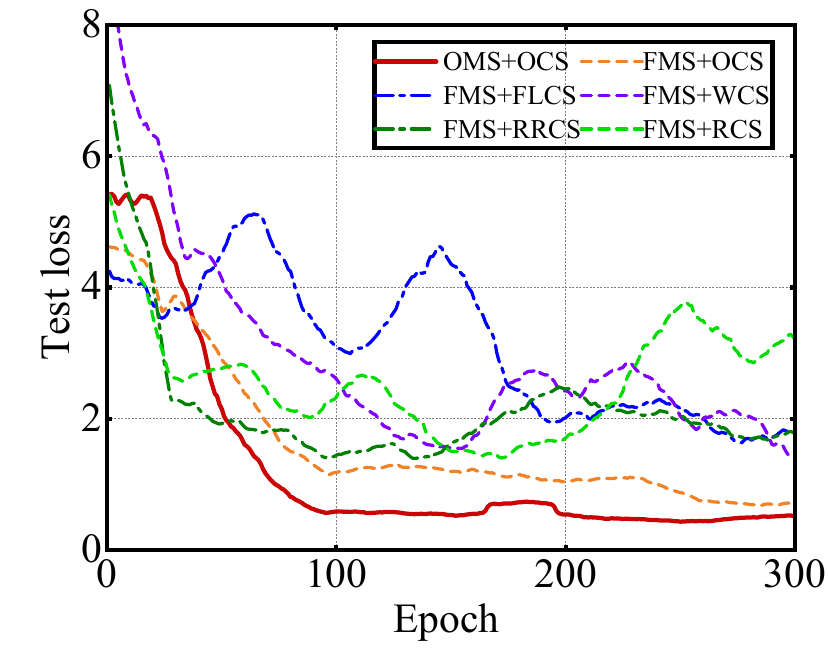}
	}
\subfigure[\centering Loss with training round (EMNIST, non-IID settings).] {
\centering\includegraphics[width=3.9cm]{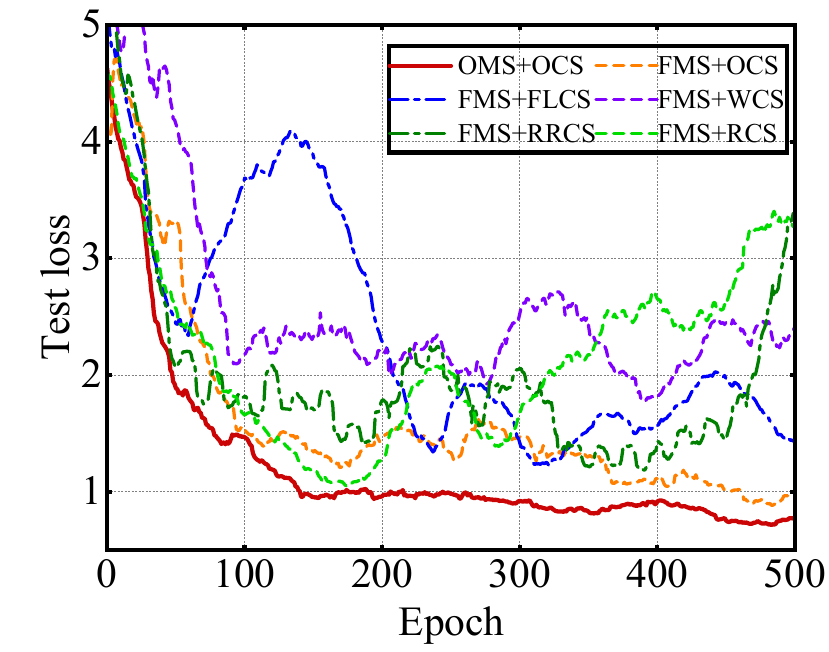}
	}
    \subfigure[\centering Loss with training round (CIFAR-10, IID settings).] {
\centering\includegraphics[width=3.9cm]{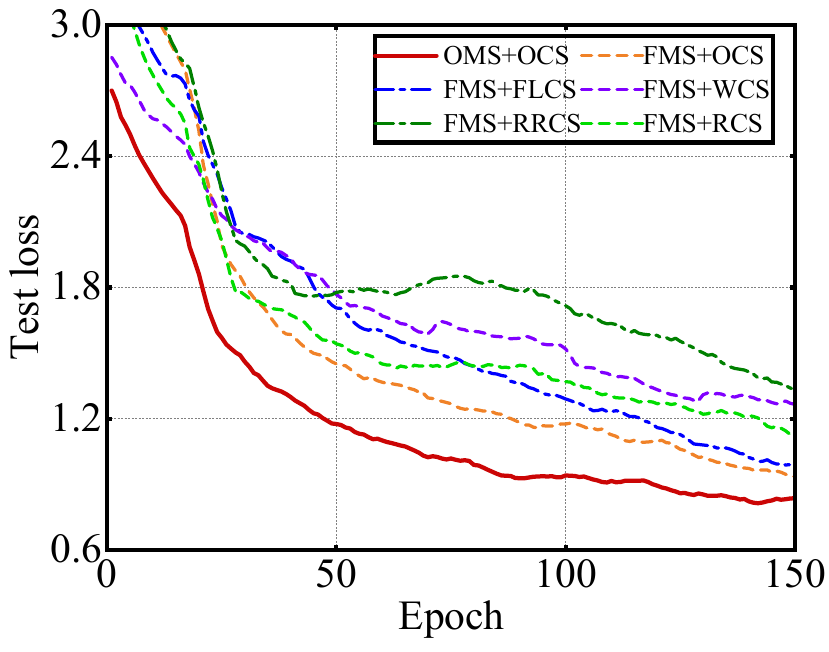}
	}
    \subfigure[\centering Loss with training round (CIFAR-10, non-IID settings).] {
\centering\includegraphics[width=3.9cm]{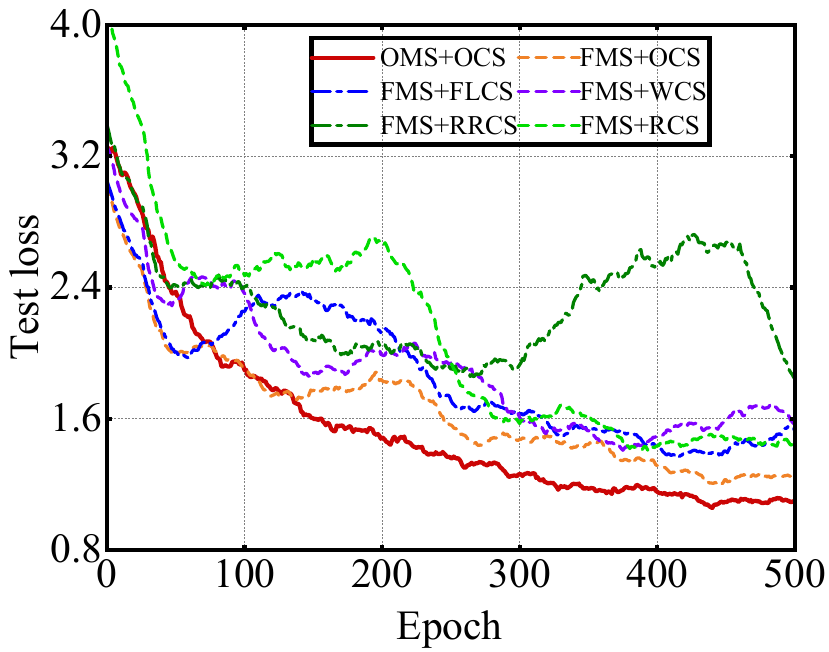}
	}
         \subfigure[\centering Accuracy with training round (EMNIST, IID settings).] {
\centering\includegraphics[width=3.9cm]{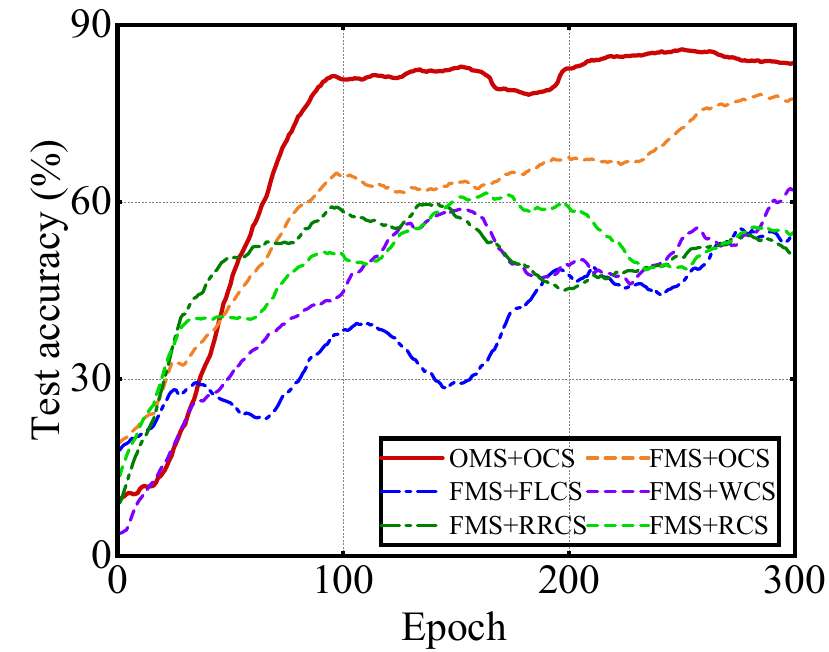}
	}
                 \subfigure[\centering Accuracy with training round (EMNIST, non-IID settings).] {
\centering\includegraphics[width=3.9cm]{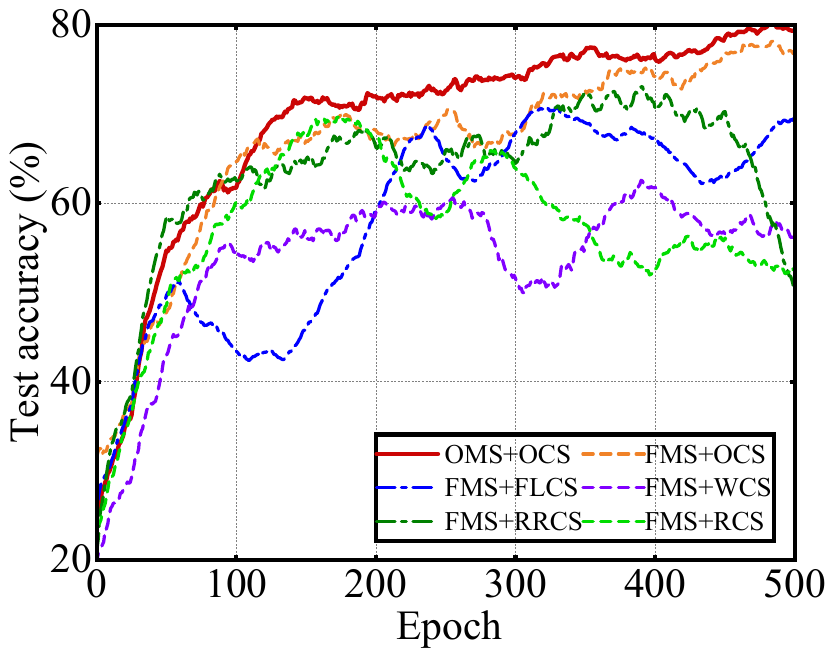}
	}
         \subfigure[\centering Accuracy with training round (CIFAR-10, IID settings).] {
\centering\includegraphics[width=3.9cm]{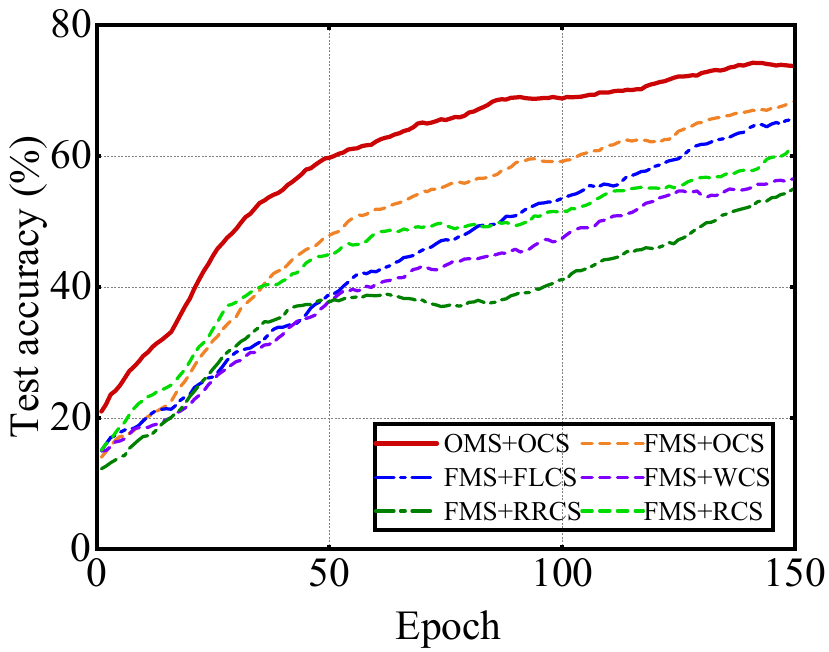}
	}
         \subfigure[\centering Accuracy with training round (CIFAR-10, non-IID settings).] {
\centering\includegraphics[width=3.9cm]{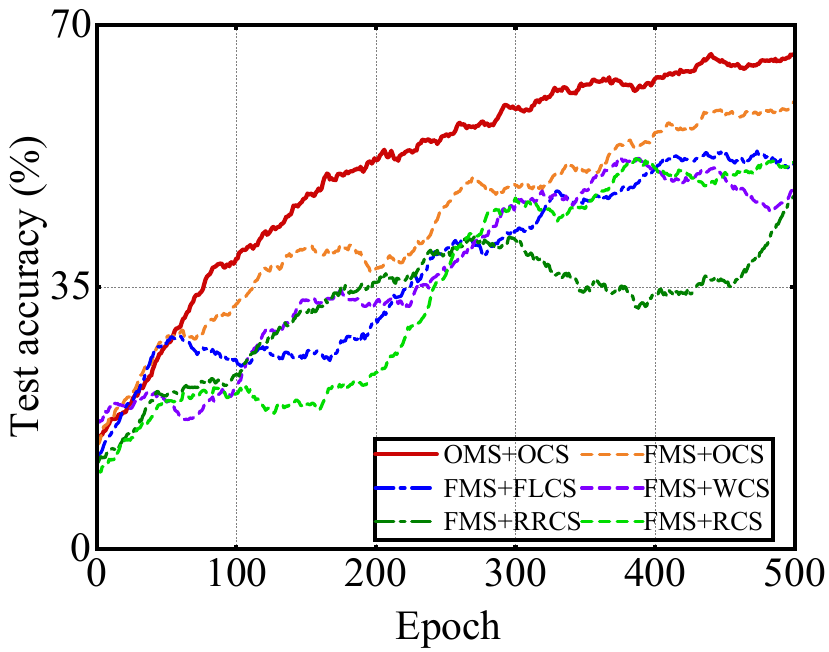}
	}
	\centering\caption{ Evaluating SFL training performance on CIFAR-10 and EMNIST using ResNet-50 with unstable clients under different model splitting and client sampling schemes ($N=100$, $K=5$,  $p_i,\varphi_i,a_i\in[0.2,0.6]$).} 
		\label{Training_curves}
        \vspace{-0.8em}
\end{figure*}

 To demonstrate the advantages of our framework,  we perform a comprehensive comparison between our proposed  \textbf{(OMS+OCS)} method
 and five benchmarks: 
 1) ``Fixed  Model Splitting and Optimal Client Sampling'' (\textbf{FMS-OCS}) 
 selects a fixed cut layer for each client at random while adopting the proposed optimal client sampling method.
 2) ``Fixed Model Splitting and FL-Based Client Sampling'' (\textbf{FMS+FLCS}) selects a fixed cut layer for each client at random, while incorporating an FL-based client sampling strategy that accounts for sample importance~\cite{chen2022optimalclientsamplingfederated}, where the importance is quantified by the norm of model updates. 
3) ``Fixed Model Splitting and Weighted Client Sampling'' (\textbf{FMS+RCS}) utilizes the fixed model splitting strategy and selects participating clients based on the size of client’s datasets.
4) ``Fixed  Model Splitting and Round-Robin Client Sampling'' (\textbf{FMS+RRCS}) employs a randomly selected fixed cut layer for model partitioning and a round-robin strategy for client sampling~\cite{zhang2025optimalheterogeneousclientsampling}.  
5) ``Fixed  Model Splitting and Random Client Sampling'' (\textbf{FMS-RCS}) leverages fixed model splitting and random client sampling.

 \subsection{Estimation of Parameters \texorpdfstring{$G_i^2$}{G\_i²}, \texorpdfstring{$\sigma_{i}^{2}$}{sigma\_i²} and \texorpdfstring{$\beta$}{beta}}
In principle, the true global smoothness constant
 $ \beta \;=\; \sup_{w,w'} \frac{\|\nabla f(w) - \nabla f(w')\|}{\|w - w'\|}$
is essentially impossible to compute for deep neural networks~\cite{8664630}: one would have to search over all pairs of parameter vectors in a space of millions of dimensions, and even local Hessian‐based methods (e.g.\ estimating the largest eigenvalue) are prohibitively expensive and valid only in an infinitesimal neighborhood of a single point.
Instead, we adopt a lightweight empirical per‐layer approximation.
First, on each client $i$, we perform a finite‐difference test on a minibatch of size \(n_i\) and  add a small perturbation \(\varepsilon\), computing
\begin{equation}
    \beta_i^{\mathrm{batch}}
  = \frac{\|\nabla f(w + \varepsilon) - \nabla f(w)\|}{\|\varepsilon\|},
\end{equation}
and further obtaining
\begin{equation}
  \beta_{\mathrm{local}}
  = \frac{\sum\nolimits_{i=1}^N \beta_i^{\mathrm{batch}}\,n_i}{\sum\nolimits_{i=1}^N n_i}\,.
\end{equation}
Second, we compare pairs of client models \(w_i,\ w_z (i,z\in[1,N],i\neq z)\) on the same data to obtain
 $
  \beta_{iz}
  = \frac{\|\nabla f(w_i) - \nabla f(w_z)\|}{\|w_i - w_z\|}
  \quad\Longrightarrow\quad
  \beta_{\mathrm{cross}} = \max_{i,z}\beta_{iz}\,.
 $
Third, for each layer \(j\), we record over one epoch the variance of its squared gradient norm,
\(\sigma_j^2=\mathrm{Var}(\|\nabla_{w_j}f\|^2)\), and its maximum squared norm \(G_j^2=\max(\|\nabla_{w_j}f\|^2)\), then form 
\(\sum_j(\sigma_j^2 + G_j^2)\) as an empirical upper bound on total gradient variation.  
Finally, we set
\begin{equation}
  \beta \;=\;\max\bigl(\beta_{\mathrm{local}},\beta_{\mathrm{cross}}\bigr)
\end{equation}
and incorporate \(\sigma_j^2,\ G_j^2\) and \(\beta\) into our model splitting and client sampling procedure.
\begin{figure}[t!]
         \subfigure[\centering The impact of $p_i$ on test accuracy ($\varphi_i,a_i=0$).] {
\centering\includegraphics[width=3.85cm]{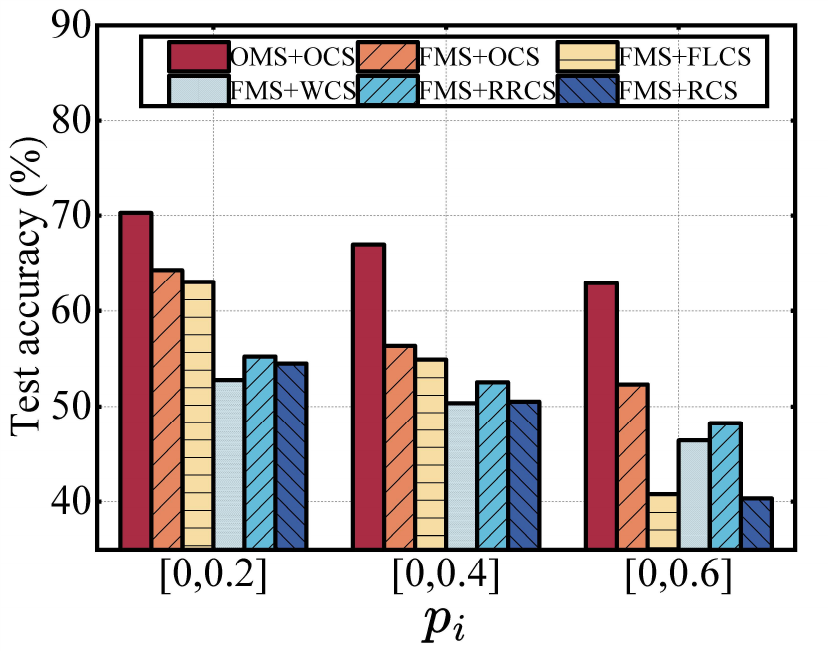}
	}
\subfigure[\centering The impact of $\varphi_i$ on test accuracy ($p_i,a_i=0$).] {
\centering\includegraphics[width=3.85cm]{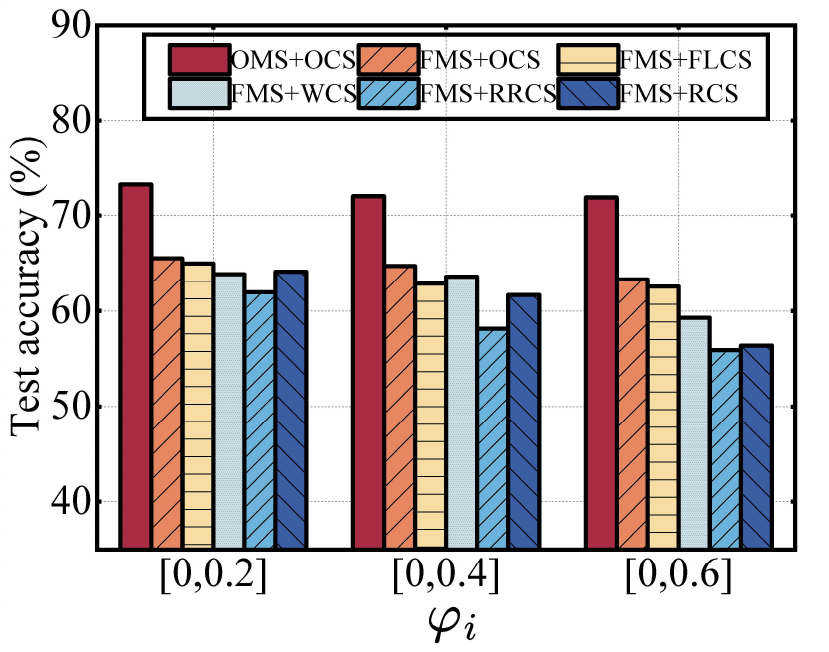}
	}
    \subfigure[\centering The impact of $a_i$ on test accuracy ($p_i,\varphi_i=0$).] {
\centering\includegraphics[width=3.85cm]{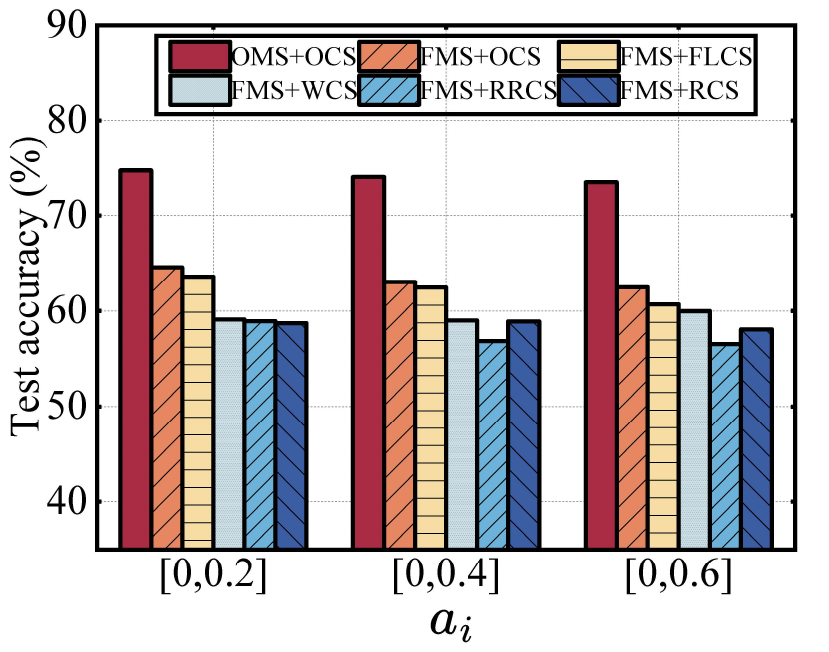}
	}
         \subfigure[\centering The impact of estimation errors in $\varpi_i$ ($\varpi_i\in\{p_i,\varphi_i,a_i\}$)
         on test accuracy. ] {
\centering\includegraphics[width=3.85cm]{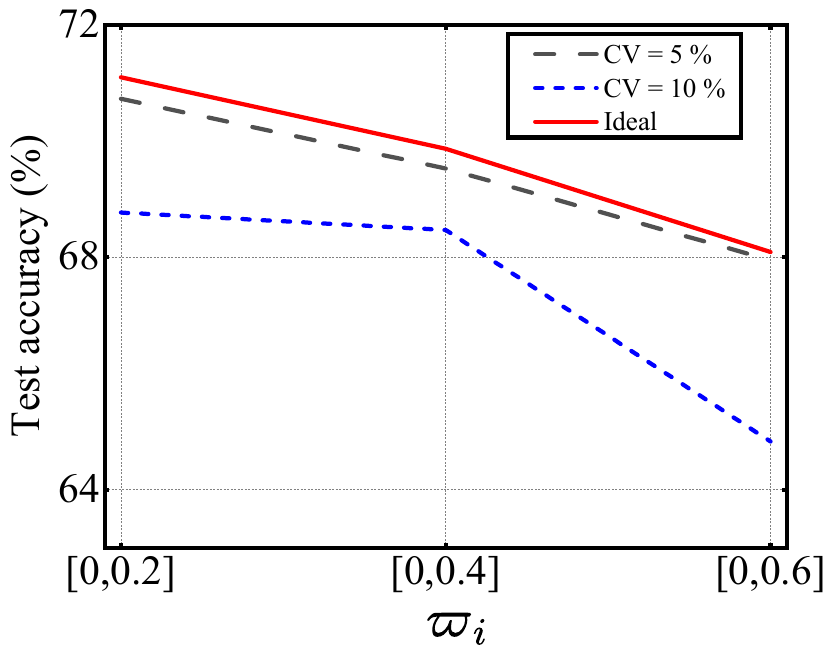}
        }
	\centering\caption{ The impact of uploading failure $p_i$, downloading failure $\varphi_i$, aggregation failure $a_i$ and estimation errors on these parameters on test accuracy.} 
		\label{pi_phi_i_ai}
        \vspace{-0.5em}
\end{figure}

 \subsection{Performance Comparison}
Fig. \ref{Training_curves} presents test accuracy and test loss versus training epochs, comparing the proposed methods with baselines on EMNIST and CIFAR-10 under both IID and non-IID data distributions. Notably, the \textbf{OMS+OCS} strategy outperforms benchmarks by exhibiting the highest test accuracy and the lowest test loss across all scenarios. 
Specifically, on the EMNIST dataset under IID settings, the proposed \textbf{OMS+OCS} strategy achieves a test accuracy of 83.48\% at 300 epochs, which is significantly higher than the benchmark performances of 77.48\%, 53.83\%, 62.20\%, 51.75\%, and 55.19\%, respectively. Under non-IID settings on EMNIST, at 500 epochs \textbf{OMS+OCS} reaches 79.29\%, outperforming the benchmarks that achieve 76.81\%, 69.68\%, 56.10\%, 50.97\%, and 52.38\% in turn.
For the CIFAR-10 dataset, under IID conditions at 150 epochs, \textbf{OMS+OCS} obtains a test accuracy of 73.74\%, compared to the benchmarks’ 68.27\%, 65.84\%, 56.58\%, 55.09\%, and 61.42\%. Similarly, under non-IID conditions at 500 epochs, our proposed approach achieves 65.98\% accuracy, clearly surpassing the benchmarks at 59.90\%, 51.25\%, 48.16\%, 47.23\%, and 51.52\%.
This superior performance is primarily due to two key factors. First, \textbf{OMS+OCS} effectively identifies and mitigates the influence of unstable client contributions during the SFL training process. In SFL, differing client behaviors or data distributions can lead to significant variance in local gradients. Our approach leverages an optimal client sampling and model splitting scheme to achieve a more stable and reliable global model update. Second, \textbf{OMS+OCS} is specifically designed to reduce the overall gradient variation by dynamically adjusting the contribution of each client. The empirical upper bound based on the variance and maximum squared norm of gradients ensures that the aggregation process remains robust even under challenging non-IID conditions.
Overall, our experimental results demonstrate that \textbf{OMS+OCS} is a highly efficient and resilient approach for addressing the challenges posed by unstable clients in SFL.



In Fig. \ref{pi_phi_i_ai} (a)-(c), to evaluate model performance under varying client drop rate scenarios, we simulate 100 clients by varying  the upload failure probability \(p_i\), the download failure probability \(\varphi_i\), and the aggregation failure probability \(a_i\), separately, across three distinct scenarios: light drop rate scenario where the probability is sampled from \([0, 0.2]\), moderate drop rate scenario where the probability is sampled from \([0, 0.4]\), and severe drop rate scenario where the probability is sampled from \([0, 0.6]\). 
The results show that as the unstable condition becomes severe, the test accuracy for each scheme decreases, emphasizing the negative impact of higher client drop rates on overall model training performance.
Moreover, the experimental results validate that the proposed  \textbf{OMS+OCS} strategy effectively mitigates negative effects of unstable clients  on model training performance. In particular, we notice that \(p_i\) has a severer negative impact on test accuracy, which is consistent with \textbf{Insight 2}.

In edge networks, network resources often fluctuate during training, resulting in discrepancies between the measured conditions (used for optimization) and the actual network conditions. To assess the effect of such measurement errors on client drop rates, we inject Gaussian noise with varying coefficients of variation (CV) into the modeled drop probabilities~\cite{10053757,yoo2024modeling}. As shown in Fig.~\ref{pi_phi_i_ai} (d), the proposed method is resilient across noise levels, exhibiting only slight performance loss at CV = 5\% and moderate degradation at CV = 10\%.  These results demonstrate that \textbf{OMS+OCS} remains effective even when predictions of these probabilities are unreliable.




\begin{figure}[t!]
         \subfigure[\centering The impact of $K$ on test accuracy ($N=100$).] {
\centering\includegraphics[width=3.85cm]{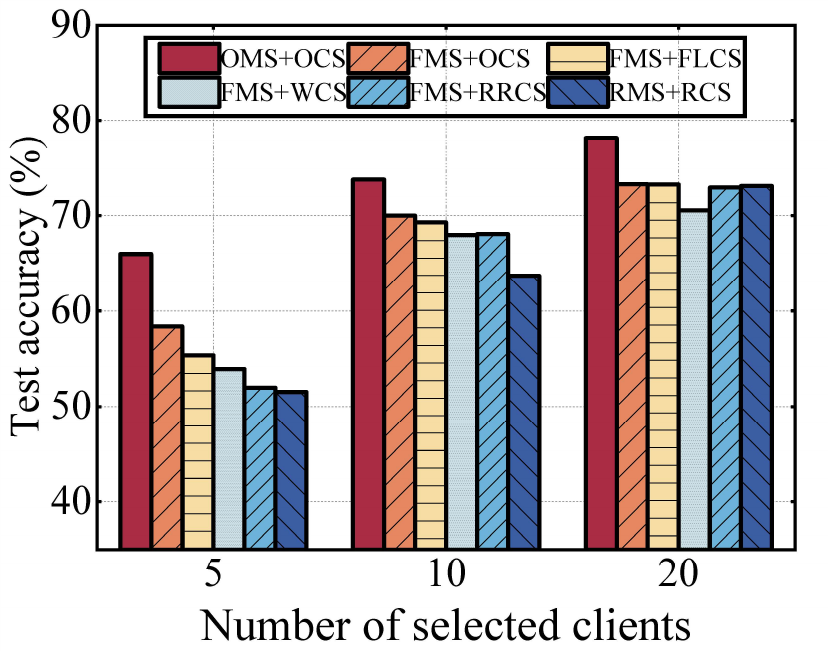}
	}
         \subfigure[\centering The impact of aggregating interval (in epochs) on test accuracy.] {
 \centering
		\centering\includegraphics[width=3.85cm]{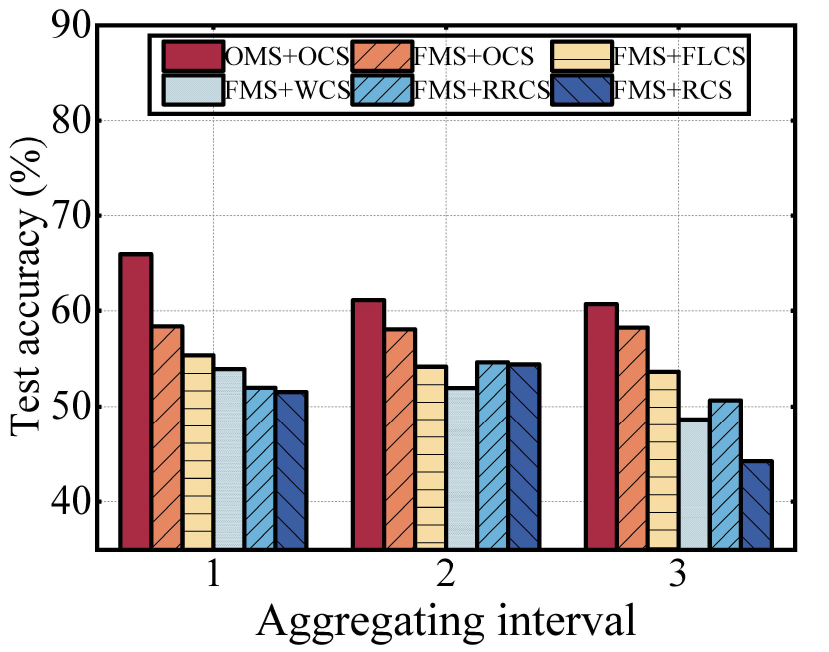}
        }
	\centering\caption{The impact of number of selected clients $K$ and aggregation interval on test accuracy.} 
		\label{number_interval}
         \vspace{-0.5em}
\end{figure}

Fig.~\ref{number_interval} (a) shows how the number of selected clients $K$ affects test accuracy. As the number of selected clients increases, the test accuracy of each scheme generally improves. This indicates that involving more clients provides the model with richer data, thereby enhancing its generalization ability.
Fig.~\ref{number_interval} (b), on the other hand, examines the impact of different aggregation intervals on test accuracy. The results demonstrate that as the aggregation interval increases, the test accuracy of each scheme tends to decline. This suggests that more frequent aggregation (i.e., shorter intervals) facilitates the timely integration of updates from each client, thereby improving the effectiveness of model training. Conversely, longer intervals may lead to delays in updating information, which can negatively affect the final model performance.

\begin{table*}[ht]
\centering
\textcolor{black}{
\caption{Converged test accuracy comparison under different model and dataset settings.}
\label{tab:test_accuracy1}
\footnotesize
\setlength{\tabcolsep}{3.5pt} 
\renewcommand{\arraystretch}{0.95} 
\begin{tabular}{c|cccccc}
\toprule
\multirow{2}{*}{Settings}
& \multicolumn{6}{c}{Converged Test Accuracy (\%)} \\
\cmidrule(lr){2-7}
& \textbf{OMS+OCS} & \textbf{FMS+OCS} & \textbf{FMS+FLCS} & \textbf{FMS+WCS} & \textbf{FMS+RRCS} & \textbf{FMS+RCS} \\
\midrule
ResNet-50 + EMNIST   & \textbf{79.29} & 76.81 & 69.68 & 56.10 & 50.97 & 52.38 \\
ResNet-50 + CIFAR-10 & \textbf{65.98} & 59.90 & 51.25 & 48.16 & 47.23 & 51.52 \\
ResNet-50 + CINIC-10 & \textbf{60.25} & 57.31 & 50.17 & 51.51 & 54.49 & 50.57 \\
\midrule
ResNet-101 + EMNIST   & \textbf{89.70} & 81.69 & 68.73 & 60.04 & 63.24 & 57.84 \\
ResNet-101 + CIFAR-10 & \textbf{64.53} & 62.53 & 61.14 & 57.37 & 58.76 & 55.05 \\
ResNet-101 + CINIC-10 & \textbf{59.31} & 56.05 & 49.63 & 51.21 & 51.54 & 52.09 \\
\bottomrule
\end{tabular}
}
\vspace{-0.8em}
\end{table*}

\textcolor{black}{
We  extend the evaluation beyond EMNIST and CIFAR-10 by including CINIC-10\cite{darlow2018cinic10imagenetcifar10}, which augments CIFAR-10 with additional downsampled ImageNet images. With multi-source construction, CINIC-10 exhibits higher intra-class diversity and a more pronounced domain shift, making optimization under non-IID client distributions more demanding. We further evaluate a \emph{deeper model architecture} (ResNet-101~\cite{he2015deepresiduallearningimage}) to increase model depth and computational complexity. As shown in Table~\ref{tab:test_accuracy1}, the proposed framework remains effective  across all evaluated datasets and model settings, achieving consistent converged test accuracy.
}

 \subsection{Ablation Study}
\begin{figure}[t!]
         \subfigure[\centering CIFAR-10 under IID setting.] {
\centering\includegraphics[width=3.85cm]{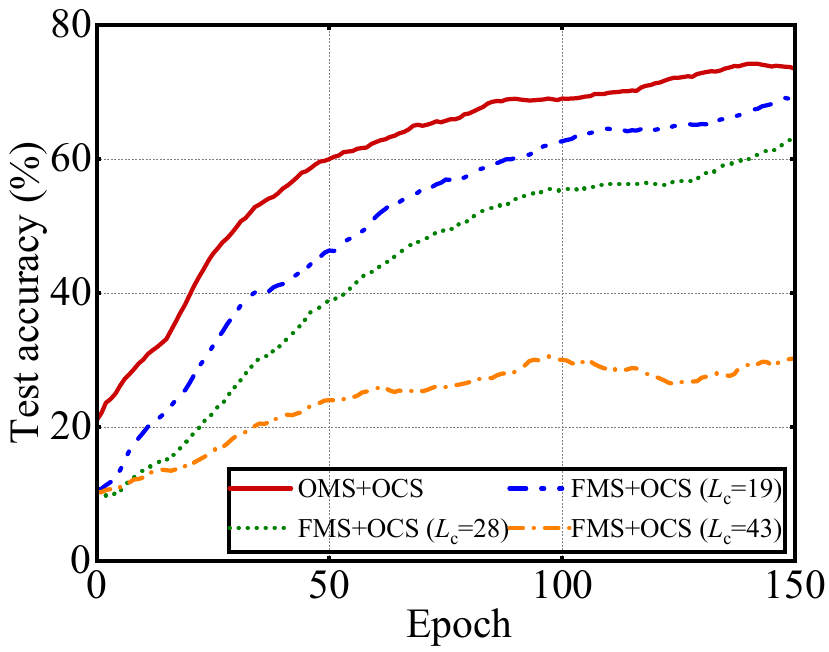}
	}
         \subfigure[\centering CIFAR-10 under non-IID setting.] {
 \centering
		\centering\includegraphics[width=3.85cm]{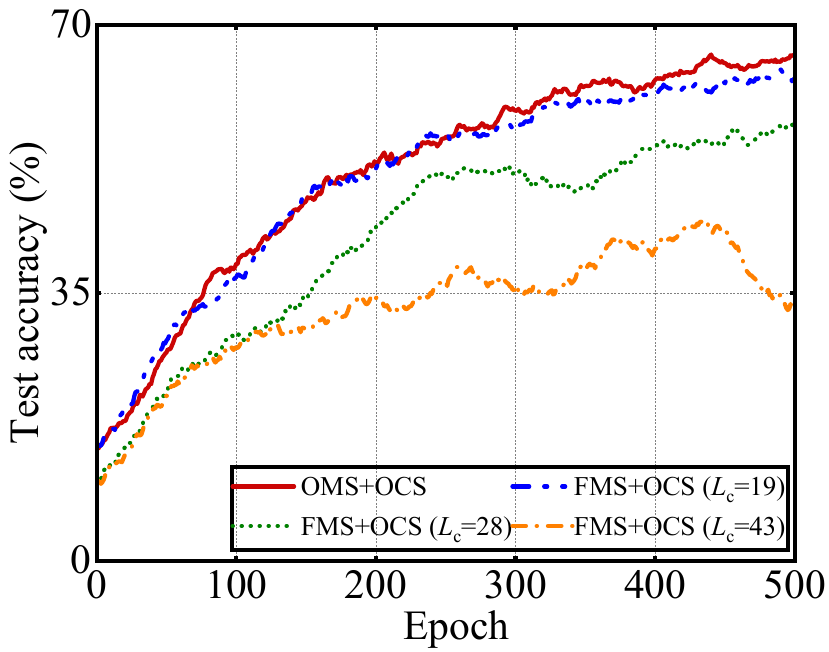}
        }
	\centering\caption{Ablation experiments for model splitting scheme on the CIFAR-10 dataset under IID and non-IID settings.} 
		\label{ablation}
        \vspace{-0.5em}
\end{figure}
Fig.~\ref{ablation} presents the training performance on the CIFAR-10 dataset under both IID and non-IID data distributions with the \textbf{OMS+OCS} and \textbf{FMS+OCS} strategies, serving as the ablation study on model splitting. We observe that the \textbf{OMS+OCS} approach achieves the highest test accuracy. In contrast, the performance of the \textbf{FMS+OCS} approach is critically dependent on the value of $L_c$. This comparison validates the effectiveness of our model splitting strategy.



\section{Conclusion}\label{conclusion}
\textcolor{black}{In this paper, we have presented an optimal framework for client sampling and model splitting in SFL under resource-constrained and unstable network conditions. By formulating a joint optimization problem that minimizes the global training loss, our approach effectively enhances training performance with heterogeneous edge devices under unstable network environments. Extensive simulations over diverse datasets have validated that our method significantly improves test accuracy compared to existing baselines. Future work will explore further enhancements to adaptively manage dynamic network resources and extend the framework to edge computing scenarios with more complex client drop patterns, and investigate online learning algorithms for highly
dynamic or partially observable environments.}

\appendices

\section{A. Proof of Lemma 1} \label{aa}
We first fix a training round \( t \geq 1 \). Let \( t_0 \leq t \) be the largest integer such that \( t_0 \bmod I = 0 \), so \( 0 \leq t - t_0 \leq I \).
Recalling the model updating rules, we have
{\begin{align}
\,\,\mathbf{h}_{c,i}^{t}&=\mathbf{h}_{c}^{t_0}-\gamma \sum_{\tau =t_0+1}^t{\frac{\mathbf{1}\!\left( {s_{a}^{i}} \right)}{\left( 1-a_i \right)}\frac{\mathbf{1}\!\left( {s_{u}^{i}},\,\,{s_{d}^{i}} \right)}{\left( 1-p_i \right) \left( 1-\varphi _i \right)}}\,\,\mathbf{g}_{c,i}^{\tau}
, \label{eq:local_update} \\
\mathbf{h}_{c}^{t}&=\mathbf{h}_{c}^{t_0}-\gamma \sum_{\tau =t_0+1}^t\sum_{i\in \mathcal{K} ^{\left( \tau -1 \right)}(\boldsymbol{q})}\frac{1}{K}\frac{m_i}{q_i}\frac{\mathbf{1}\!\left( {s_{a}^{i}} \right)}{\left( 1-a_i \right)}\cdot\notag\\
&\quad \quad \quad \quad \quad \ \quad \quad \quad  \quad  \quad \quad \frac{\mathbf{1}\!\left( {s_{u}^{i}},\,\,{s_{d}^{i}} \right)}{\left( 1-p_i \right) \left( 1-\varphi _i \right)}\mathbf{g}_{c,i}^{\tau}.
 \label{eq:global_update}
\end{align}
}


We bound the expected squared norm
\begin{align}
&\mathbb{E} \left[ \| \mathbf{h}_{c}^{t} - \mathbf{h}_{c,i}^{t} \|^{2} \right] \leq \gamma^{2} \left( t - t_0 \right)^2 \notag \\&\quad \quad \quad \quad \quad \sum_{\tau = t_0 + 1}^{t} \mathbb{E} \bigg[ \bigg\| \sum_{i \in \mathcal{K}^{(\tau -1)}} A_{i} + D_{i} \bigg\|^{2} \bigg], \label{eq:bound1}
\end{align}
where
{\begin{align}
A_{i} &= 
\frac{1}{K}\frac{m_i}{q_i}\frac{\mathbf{1}\!\left( {s_{a}^{i}} \right)}{\left( 1-a_i \right)}\frac{\mathbf{1}\!\left( {s_{u}^{i}},\,\,{s_{d}^{i}} \right)}{\left( 1-p_i \right) \left( 1-\varphi _i \right)} \mathbf{g}_{c,i}^{\tau}
, \notag \\
D_{i} &= 
\frac{\mathbf{1}\!\left( {s_{a}^{i}} \right)}{\left( 1-a_i \right)}\frac{\mathbf{1}\!\left( {s_{u}^{i}},\,\,{s_{d}^{i}} \right)}{\left( 1-p_i \right) \left( 1-\varphi _i \right)}
 \mathbf{g}_{c,i}^{\tau}.
\end{align}}

Using the inequality \( \| \mathbf{a} + \mathbf{b} \|^{2} \leq 2 (\| \mathbf{a} \|^{2} + \| \mathbf{b} \|^{2}) \), it follows that
{\begin{align}
&\mathbb{E} \bigg[ \bigg\| \sum_{i \in \mathcal{K}^{(\tau -1)}(\boldsymbol{q})} A_{i} + D_{i} \bigg\|^{2} \bigg] \notag\\&\quad \quad \leq 2 \mathbb{E} \bigg[ \bigg\| \sum_{i \in \mathcal{K}^{(\tau -1)}(\boldsymbol{q})} A_{i} \bigg\|^{2} \bigg] + 2 \mathbb{E} \left[ \| D_{i} \|^{2} \right]. \label{eq:bound2}
\end{align}}

Thus, we have
{\begin{align}
&\mathbb{E} \bigg[ \bigg\| \sum_{i \in \mathcal{K}^{(\tau -1)}(\boldsymbol{q})} A_{i}  \bigg\|^{2} \bigg] \notag \\
&\overset{(a)}{\leq } \frac{1}{K}\mathbb{E} _{}[\sum_{i\in \mathcal{K} ^{\left( \tau -1 \right)}(\boldsymbol{q})}{\frac{{m_i}^2}{{q_i}^2}\frac{\mathbf{1}\!\left( {s_{a}^{i}} \right) ^2}{\left( 1-a_i \right) ^2}\frac{\mathbf{1}\!\left( {s_{u}^{i}},\,\,{s_{d}^{i}} \right) ^2}{\left( 1-p_i \right) \left( 1-\varphi _i \right) ^2}}\parallel \mathbf{g}_{c,i}^{\tau}\parallel ^2]
\notag\\
& \overset{(b)}{=}\mathbb{E} _{i\in \mathcal{K} ^{\left( \tau -1 \right)}(\boldsymbol{q})}[\frac{{m_i}^2}{{q_i}^2}\frac{\mathbf{1}\!\left( {s_{a}^{i}} \right) ^2}{\left( 1-a_i \right) ^2}\frac{\mathbf{1}\!\left( {s_{u}^{i}},\,\,{s_{d}^{i}} \right) ^2}{\left( 1-p_i \right) \left( 1-\varphi _i \right) ^2}\parallel \mathbf{g}_{c,i}^{\tau}\parallel ^2]
\notag\\
&\overset{(c)}{\leq } \sum_{i=1}^N{q_i}\frac{{m_i}^2}{{q_i}^2}\frac{1}{\left( 1-a_i \right)}\frac{1}{\left( 1-p_i \right) \left( 1-\varphi _i \right)}\sum_{j=1}^{L_c}{G_{j}^{2}}
, \label{eq:bound3}
\end{align}
where $(a)$ follows by using $
||\sum\nolimits_{i=1}^n{z_i||^2\le}n\sum\nolimits_{i=1}^n{||z_i||^2}
$; $(b)$ follows the proof of Lemma 1 in 
\cite{10443546};
and $(c)$ follows from \textbf{Assumption 2}.}

Similarly, we have
{\begin{align}
\mathbb{E} \left[ \parallel D_i\parallel ^2 \right] \le \frac{1}{\left( 1-a_i \right)}\frac{1}{\left( 1-p_i \right) \left( 1-\varphi _i \right)}\sum_{j=1}^{L_c}{G_{j}^{2}}. 
\label{eq:bound4}
\end{align}}

Combining \eqref{eq:bound2}, \eqref{eq:bound3}, and \eqref{eq:bound4}, and noting \( t - t_0 \leq I \), we have
{\begin{align}
&\mathbb{E} \left[ \| \mathbf{h}_{c}^{t} - \mathbf{h}_{c,i}^{t} \|^{2} \right] \notag\\&\leq 
2\gamma ^2I^2\bigg(\sum_{i=1}^N{}\frac{{m_i}^2}{q_i}\frac{1}{\left( 1-a_i \right)}\frac{1}{\left( 1-p_i \right) \left( 1-\varphi _i \right)}+\notag\\
&\quad \quad \quad \quad \quad \quad \frac{1}{\left( 1-a_i \right)}\frac{1}{\left( 1-p_i \right) \left( 1-\varphi _i \right)}\bigg)\sum_{j=1}^{L_c}{G_{j}^{2}}, \notag \\
&\leq 2\gamma ^2I^2\bigg(N\underset{i^\prime}{\max}\big\{ \frac{{m_{i^\prime}}^2}{q_{i^\prime}}\frac{1}{\left( 1-a_{i^\prime} \right)}\frac{1}{\left( 1-p_{i^\prime} \right) \left( 1-\varphi_{i^\prime} \right)} \big\} +\notag \\
&\quad \quad \quad \quad \quad \quad\frac{1}{\left( 1-a_i \right)}\frac{1}{\left( 1-p_i \right) \left( 1-\varphi _i \right)}\bigg)\sum_{j=1}^{L_c}{G_{j}^{2}}.
\label{eq:final_bound}
\end{align}}

This completes the proof.

\section{B. Proof of Theorem 1}\label{bb}
For all training rounds  $t \ge 1$,  the smoothness of the loss function $f(\cdot)$ implies that
{
\begin{align}\label{eq:equal_1_total}
\mathbb{E}[ {f({{\bf{w}}^t})} ] \le& {\rm{ }}\mathbb{E}[ {f({{\bf{w}}^{t - 1}})} ] + \frac{\beta }{2}\mathbb{E}[ {{{\| {{{\bf{w}}^t} - {{\bf{w}}^{t - 1}}} \|}^2}} ] \nonumber \\&~+ 
\mathbb{E}[ {\langle {\nabla _{\bf{w}}}f({{\bf{w}}^{t - 1}}),{{\bf{w}}^t} - {{\bf{w}}^{t - 1}}\rangle } ]{\rm{  }}
.
\end{align}
}

According to \cite{lin2024adaptsfl}, it follows that
{\begin{align}\label{eq:w_decouple}
 \mathbb{E}[ {{{\| {{{\bf{w}}^t} - {{\bf{w}}^{t - 1}}} \|}^2}} ]=\mathbb{E}[ {{{\| { {{\bf{h}}_c^t - {\bf{h}}_c^{t - 1}} } \|}^2}} ] + \mathbb{E}[ {{{\| { {{\bf{h}}_s^t - {\bf{h}}_s^{t - 1}} } \|}^2}} ],
\end{align}}%
where $\mathbb{E}[ {{{\| { {{\bf{h}}_c^t - {\bf{h}}_c^{t - 1}} } \|}^2}} ]$ can be bounded as 

{
\begin{align}\label{eq:wc_squre}
 &\mathbb{E}[ {{{\| { {{\bf{h}}_c^t - {\bf{h}}_c^{t - 1}} } \|}^2}} ] \notag\\ \overset{(a)}{=}&
\gamma ^2\mathbb{E} [||\frac{1}{K}\sum_{i\in \mathcal{K} ^{\left( t-1 \right)}(\boldsymbol{q})}\frac{m_i}{q_i}{\frac{\mathbf{1}\!\left( {s_{a}^{i}} \right)}{\left( 1-a_i \right)}}\frac{\mathbf{1}\!\left( {s_{u}^{i}},\,\,{s_{d}^{i}} \right)}{\left( 1-p_i \right) \left( 1-\varphi _i \right)}\mathbf{g}_{c,i}^{t}||^2]
 \nonumber\\ 
 \overset{(b)}{=}& {\gamma ^2}\mathbb{E}[||
\frac{1}{K}\sum_{i\in \mathcal{K} ^{\left( t-1 \right)}(\boldsymbol{q})}\frac{m_i}{q_i}{\frac{\mathbf{1}\!\left( {s_{a}^{i}} \right)}{\left( 1-a_i \right)}}\frac{\mathbf{1}\!\left( {s_{u}^{i}},\,\,{s_{d}^{i}} \right)}{\left( 1-p_i \right) \left( 1-\varphi _i \right)}
 {( \bf{g}}_{c,i}^t - \nonumber\\&{\nabla _{{{\bf{h}}_c}}}{f_i}\left( {{\bf{h}}_{c,i}^{t - 1}} \right) ) |{|^2}] +{\gamma ^2}\mathbb{E}[||\frac{1}{K}\sum_{i\in \mathcal{K} ^{\left( t-1 \right)}(\boldsymbol{q})}\frac{m_i}{q_i}\frac{\mathbf{1}\!\left( {s_{a}^{i}} \right)}{\left( 1-a_i \right)}\cdot\nonumber\\ 
& 
\frac{\mathbf{1}\!\left( {s_{u}^{i}},\,\,{s_{d}^{i}} \right)}{\left( 1-p_i \right) \left( 1-\varphi _i \right)}
 {{\nabla _{{{\bf{h}}_c}}}{f_i}\left( {{\bf{h}}_{c,i}^{t - 1}} \right)} |{|^2}] \nonumber\\
\overset{(c)} \leq & 
\frac{\gamma ^2}{K}\mathbb{E} [\sum_{i\in \mathcal{K} ^{(t-1)}(\boldsymbol{q})} {\frac{{m_i}^2}{{q_i}^2}}{\frac{\mathbf{1}\!\left( {s_{a}^{i}} \right) ^2}{\left( 1-a_i \right) ^2}}\frac{\mathbf{1}\!\left( {s_{u}^{i}},\,\,{s_{d}^{i}} \right) ^2}{\left( 1-p_i \right) ^2\left( 1-\varphi _i \right) ^2} \cdot
\nonumber\\
 & 
||\mathbf{g}_{c,i}^{t}-\nabla _{\mathbf{h}_c}f_i\left( \mathbf{h}_{c,i}^{t-1} \right) ||^2]+\gamma ^2\mathbb{E} [||\sum_{i\in \mathcal{K} ^{\left( t-1 \right)}(\boldsymbol{q})}{\frac{1}{K}\frac{m_i}{q_i}}\cdot\nonumber\\
&\frac{\mathbf{1}\!\left( {s_{a}^{i}} \right)}{\left( 1-a_i \right)}\frac{\mathbf{1}\!\left( {s_{u}^{i}},\,\,{s_{d}^{i}} \right)}{\left( 1-p_i \right) \left( 1-\varphi _i \right)}\nabla _{\mathbf{h}_c}f_i\left( \mathbf{h}_{c,i}^{t-1} \right) ||^2]
 \nonumber\\
 \overset{(d)}{=}&\frac{\gamma ^2K}{K}\mathbb{E} _{i\in \mathcal{K} ^{\left( t-1 \right)}(\boldsymbol{q})}[\frac{{m_i}^2}{{q_i}^2}\frac{\mathbf{1}\!\left( {s_{a}^{i}} \right) ^2}{\left( 1-a_i \right) ^2}\frac{\mathbf{1}\!\left( {s_{u}^{i}},\,\,{s_{d}^{i}} \right) ^2}{\left( 1-p_i \right) ^2\left( 1-\varphi _i \right) ^2} \cdot \nonumber
\\
&||\mathbf{g}_{c,i}^{t}-\nabla _{\mathbf{h}_c}f_i\left( \mathbf{h}_{c,i}^{t-1} \right) ||^2]+\frac{\gamma ^2K}{K}\mathbb{E} _{i\in \mathcal{K} ^{\left( t-1 \right)}(\boldsymbol{q})}[\frac{{m_i}^2}{{q_i}^2}\cdot\nonumber\\
&\frac{\mathbf{1}\!\left( {s_{a}^{i}} \right) ^2}{\left( 1-a_i \right) ^2}\frac{\mathbf{1}\!\left( {s_{u}^{i}},\,\,{s_{d}^{i}} \right) ^2}{\left( 1-p_i \right) ^2\left( 1-\varphi _i \right) ^2}||\nabla _{\mathbf{h}_c}f_i\left( \mathbf{h}_{c,i}^{t-1} \right) ||^2] \nonumber \\
\overset{(e)}{\leq }& 
\gamma ^2\sum_{i=1}^N{q_i}\frac{{m_i}^2}{{q_i}^2}\frac{1}{\left( 1-a_i \right)}\frac{1}{\left( 1-p_i \right) \left( 1-\varphi _i \right)}\sum_{j=1}^{L_c}{\sigma _{j}^{2}}+\gamma ^2\sum_{i=1}^N{q_i} \cdot
 \nonumber\\
 & 
\frac{{m_i}^2}{{q_i}^2}\frac{1}{\left( 1-a_i \right)}\frac{1}{\left( 1-p_i \right) \left( 1-\varphi _i \right)}\mathbb{E} [||\nabla _{\mathbf{h}_c}f_i\left( \mathbf{h}_{c,i}^{t-1} \right) ||^2]
,
\end{align}}%
where (a) is derived using $      \mathbf{h}_{c}^{t}
      \;=\;
      \frac{1}{K}
      \sum_{i\in\mathcal{K}^{(t)}(\boldsymbol{q})}
      \frac{m_i}{q_i}\,
      \mathbf{h}_{c,i}^t$ and Eq.~\eqref{eq:local_update}; (b) follows from $\mathbb{E}[\mathbf{g}_{c,i}^{t}] = \nabla_{{\bf{h}}_c} f_{i}({\mathbf{h}}_{c,i}^{t-1})$ and the decomposition $\mathbb{E}[\Vert \mathbf{z} \Vert^{2}] = \mathbb{E} [ \Vert \mathbf{\mathbf{z}} - \mathbb{E}[\mathbf{z}]\Vert^{2}] + \Vert\mathbb{E}[\mathbf{z}] \Vert^{2}$, valid for any random vector $\mathbf{z}$; (c) applies by using $
||\sum\nolimits_{i=1}^n{z_i||^2\le}n\sum\nolimits_{i=1}^n{||z_i||^2}
$; (d) follows the proof of \textbf{Lemma 1} in \cite{10443546}; and (e) follows from {\bf Assumption 2}.

Similarly, $\mathbb{E}[ {{{\| { {{\bf{h}}_s^t - {\bf{h}}_s^{t - 1}} ]} \|}^2}} $ has an upper bound:
{
 \begin{align}\label{eq:ws_squre}
&\mathbb{E}[ {{{\| { {{\bf{h}}_s^t - {\bf{h}}_s^{t - 1}} } \|}^2}}] \le 
\gamma ^2\sum_{i=1}^N{}\frac{{m_i}^2}{q_i}\frac{1}{\left( 1-p_i \right)}\sum_{j=L_c+1}^L{\sigma _{j}^{2}}
 \nonumber\\
&
+\gamma ^2\sum_{i=1}^N{}\frac{{m_i}^2}{q_i}\frac{1}{\left( 1-p_i \right)}\mathbb{E} [||\nabla _{\mathbf{h}_s}f_i\left( \mathbf{h}_{s,i}^{t-1} \right) ||^2]
.
\end{align}}

Substituting Eq.~\eqref{eq:wc_squre} and Eq.~\eqref{eq:ws_squre} into Eq.~\eqref{eq:w_decouple} yields
{\begin{align}\label{eq:w-difference-squre}
&\mathbb{E}[\|{{\bf{w}}^t} - {{\bf{w}}^{t - 1}}\|{^2}] \nonumber \\
\le& 
\gamma ^2\sum_{i=1}^N{}\frac{{m_i}^2}{q_i}(\frac{1}{\left( 1-a_i \right)}\frac{1}{\left( 1-p_i \right) \left( 1-\varphi _i \right)}\sum_{j=1}^{L_c}{\sigma _{j}^{2}} +\frac{1}{\left( 1-p_i \right)}
\nonumber \\
&  
\sum_{j=L_c+1}^L{\sigma _{j}^{2}})+\gamma ^2\sum_{i=1}^N{}\frac{{m_i}^2}{q_i}\frac{1}{\left( 1-p_i \right)}\mathbb{E} [||\nabla _{\mathbf{h}_s}f_i\left( \mathbf{h}_{s,i}^{t-1} \right) ||^2]+ \nonumber \\
& 
\gamma ^2\sum_{i=1}^N{}\frac{{m_i}^2}{q_i}\frac{1}{\left( 1-a_i \right)}\frac{1}{\left( 1-p_i \right) \left( 1-\varphi _i \right)}\mathbb{E} [||\nabla _{\mathbf{h}_c}f_i\left( \mathbf{h}_{c,i}^{t-1} \right) ||^2].
\end{align}}

We further observe that 
\begin{align} \label{eq:inner_product_w}
&\mathbb{E}[\langle \nabla_{\bf{w}} f({\mathbf{w}}^{t-1}), {\mathbf{w}}^{t} - {\mathbf{w}}^{t-1}\rangle] \nonumber\\
\overset{(a)}{=}& -\gamma \mathbb{E} [\langle \nabla_{\bf{w}} f({\mathbf{w}}^{t-1}), \frac{1}{K}
\frac{m_i}{q_i}
\sum\limits_{i \in \mathcal{K}^{(t -1)}(\bm q)}  \left[ {{\bf{g}}_{s,i}^t; {\bf{g}}_{c,i}^t} \right]\rangle] \nonumber \\
\overset{(b)}{=}& -\gamma \mathbb{E}[\langle {\nabla_{\bf{w}}}f({{\bf{w}}^{t - 1}}), \frac{1}{K}\frac{m_i}{q_i}\sum\limits_{i \in \mathcal{K}^{(t -1)}(\bm q)} \nabla_{\bf{w}}  {f_i}({\bf{w}}_i^{t - 1})\rangle ]\nonumber \\
\overset{(c)}=&  - \frac{\gamma }{2}\mathbb{E}[||{\nabla _{\bf{w}}}f({{\bf{w}}^{t - 1}})|{|^2} + ||\sum\limits_{i \in \mathcal{K}^{(t -1)}(\bm q)}\frac{m_i}{q_i} {{\nabla _{\bf{w}}}} {f_i}({\bf{w}}_{i}^{t - 1})|{|^2} \cdot\nonumber \\
&  \frac{1}{K^2}- ||{\nabla _{{{\bf{w}}}}}f({{\bf{w}}^{t - 1}}) -  \frac{1}{K}\frac{m_i}{q_i}\sum\limits_{i \in \mathcal{K}^{(t -1)}(\bm q)}  {{\nabla _{{{\bf{w}}}}}} {f_i}({\bf{w}}_{i}^{t - 1})|{|^2}],
\end{align}
where (a) follows from ${{\bf{w}}^t} = \frac{1}{K}\sum\limits_{i \in \mathcal{K}^{(t -1)}(\bm q)} {\bf{w}}_i^t$; (c) follows from the  identity $\langle \mathbf{z}_{1}, \mathbf{z}_{2}\rangle = \frac{1}{2} \big( \Vert \mathbf{z}_{1}\Vert^{2} + \Vert \mathbf{z}_{2}\Vert^{2} - \Vert \mathbf{z}_{1} - \mathbf{z}_{2}\Vert^{2} \big)$ for any two vectors $\mathbf{z}_{1}, \mathbf{z}_{2}$ of the same length; (b) follows from 
	\begin{align}
	&\mathbb{E}[\langle \nabla_{\bf{w}} f({\mathbf{w}}^{t-1}), \frac{1}{K}
\frac{m_i}{q_i}
\sum\limits_{i \in \mathcal{K}^{(t -1)}(\bm q)}  \left[ {{\bf{g}}_{s,i}^t; {\bf{g}}_{c,i}^t} \right]\rangle] \notag\\
	=& \mathbb{E}[\mathbb{E}[\langle \nabla_{\bf{w}} f({\mathbf{w}}^{t-1}),  \frac{1}{K}
\frac{m_i}{q_i}
\sum\limits_{i \in \mathcal{K}^{(t -1)}(\bm q)}  \left[ {{\bf{g}}_{s,i}^t; {\bf{g}}_{c,i}^t} \right]\rangle | \boldsymbol{\xi}^{[t-1]}]]  \notag\\
	=& \mathbb{E}[\langle \nabla_{\bf{w}} f({\mathbf{w}}^{t-1}),  \frac{1}{K}
\frac{m_i}{q_i}
\sum\limits_{i \in \mathcal{K}^{(t -1)}(\bm q)} \mathbb{E}[ \left[ {{\bf{g}}_{s,i}^t; {\bf{g}}_{c,i}^t} \right]| \boldsymbol{\xi}^{[t-1]}]\rangle ] \notag\\
	 =& \mathbb{E}[\langle \nabla_{\bf{w}} f({\mathbf{w}}^{t-1}), \frac{1}{K}\frac{m_i}{q_i}\sum\limits_{i \in \mathcal{K}^{(t -1)}(\bm q)}  \nabla_{\bf{w}} f_{i}(\mathbf{w}_{i}^{t-1})\rangle ],
	\end{align}
where the first equality uses the law of expectations, the second equality holds since ${\mathbf{w}}^{t-1}$ is measurable with respect to $\boldsymbol{\xi}^{[t-1]}= [\boldsymbol{\xi}^{1}, \ldots, \boldsymbol{\xi}^{t-1}]$ and the third equality is derived using $\mathbb{E}[\mathbf{g}_{i}^{t} | \boldsymbol{\xi}^{[t-1]}] = \mathbb{E}[\nabla F_{i}(\mathbf{w}_{i}^{t-1};\xi^{t}_{i}) | \boldsymbol{\xi}^{[t-1]}] = \nabla f_{i}(\mathbf{w}_{i}^{t-1})$.

Substituting Eq.~\eqref{eq:w-difference-squre} and Eq.~\eqref{eq:inner_product_w} into Eq.~\eqref{eq:equal_1_total}, we have  
{\begin{align}
&\mathbb{E} [f(\mathbf{w}^t)]\notag\\
\overset{(a)}{\le}&\mathbb{E} [f(\mathbf{w}^{t-1})]-\frac{\gamma}{2}\mathbb{E} \left[ ||\nabla _{\mathbf{w}}f(\mathbf{w}^{t-1})||^2 \right] -\frac{\gamma}{2}\mathbb{E} [ ||\frac{1}{K}\frac{m_i}{q_i} \notag \\
&
\sum_{i\in \mathcal{K} ^{(t-1)}(\boldsymbol{q})}{\nabla _{\mathbf{w}}}f_i(\mathbf{w}_{i}^{t-1})||^2 ] +\frac{\beta}{2}\gamma ^2\sum_{i=1}^N{}\frac{{m_i}^2}{q_i}(\frac{1}{\left( 1-p_i \right) }\cdot
\notag \\
&\frac{1}{\left( 1-\varphi _i \right)}\frac{1}{\left( 1-a_i \right)}\sum_{j=1}^{L_c}{\sigma _{j}^{2}}+\frac{1}{\left( 1-p_i \right)}\sum_{j=L_c+1}^L{\sigma _{j}^{2}})+\frac{\beta}{2}\gamma ^2\cdot
\notag \\
&
\sum_{i=1}^N{}\frac{{m_i}^2}{q_i}\frac{1}{\left( 1-a_i \right)}\frac{1}{\left( 1-p_i \right) \left( 1-\varphi _i \right)}\mathbb{E} [||\nabla _{\mathbf{h}_c}f_i\left( \mathbf{h}_{c,i}^{t-1} \right) ||^2]
\notag \\
&+\frac{\beta}{2}\gamma ^2\sum_{i=1}^N{}\frac{{m_i}^2}{q_i}\frac{1}{\left( 1-p_i \right)}\mathbb{E} [||\nabla _{\mathbf{h}_s}f_i\left( \mathbf{h}_{s,i}^{t-1} \right) ||^2] \notag\\
&
+\frac{\gamma}{2}\mathbb{E} [ ||\nabla _{\mathbf{h}_s}f(\mathbf{h}_{s}^{t-1})-\frac{1}{K}\frac{m_i}{q_i}\sum_{i\in \mathcal{K} ^{(t-1)}(\boldsymbol{q})}{\nabla _{\mathbf{h}_s}}f_i(\mathbf{h}_{s,i}^{t-1})||^2 ] 
\notag\\
&+\frac{\gamma}{2}\mathbb{E} [ ||\nabla _{\mathbf{h}_c}f(\mathbf{h}_{c}^{t-1})-\frac{1}{K}\frac{m_i}{q_i}\sum_{i\in \mathcal{K} ^{(t-1)}(\boldsymbol{q})}{\nabla _{\mathbf{h}_c}}f_i(\mathbf{h}_{c,i}^{t-1})||^2 ] \notag\\
\overset{(b)}{\le} &\mathbb{E} [f(\mathbf{w}^{t-1})]-\frac{\gamma}{2}\mathbb{E} \left[ ||\nabla _{\mathbf{w}}f(\mathbf{w}^{t-1})||^2 \right] 
 \notag \\
&
+\frac{\beta}{2}\gamma ^2\sum_{i=1}^N{}\frac{{m_i}^2}{q_i}(\frac{1}{\left( 1-p_i \right) }\frac{1}{\left( 1-\varphi _i \right)}\frac{1}{\left( 1-a_i \right)}\sum_{j=1}^{L_c}{\sigma _{j}^{2}}+
\notag \\
&\frac{1}{\left( 1-p_i \right)}\sum_{j=L_c+1}^L{\sigma _{j}^{2}})+\frac{\beta}{2}\gamma ^2\sum_{i=1}^N{}\frac{{m_i}^2}{q_i}\frac{1}{\left( 1-p_i \right)}
\sum_{j=L_c+1}^L{G_j^2}
\notag \\
&
+\frac{\beta}{2}\gamma ^2\sum_{i=1}^N{}\frac{{m_i}^2}{q_i}\frac{1}{\left( 1-a_i \right)}\frac{1}{\left( 1-p_i \right) \left( 1-\varphi _i \right)}
\sum_{j=1}^{L_c}{G_j^2}
\notag \\
&
+\frac{\gamma}{2}\sum_{i=1}^N{}\beta ^2\frac{{m_i}^2}{q_i}2\gamma ^2I^2 \big( N\underset{i^\prime}{\max}\big\{ \frac{{m_{i^\prime}}^2}{q_{i^\prime}}\frac{1}{\left( 1-p_{i^\prime} \right) \left( 1-\varphi_{i^\prime} \right)}  \notag\\
&\frac{1}{\left( 1-a_{i^\prime} \right)}\big\}+\frac{1}{\left( 1-a_i \right)}\frac{1}{\left( 1-p_i \right) \left( 1-\varphi _i \right)} \big) \sum_{j=1}^{L_c}{G_{j}^{2}} 
    \label{eq:11}
\end{align}}
where (a) is justified by
\begin{align}
&\mathbb{E} [ ||\nabla _{\mathbf{w}}f(\mathbf{w}^{t-1})-\frac{1}{K}\frac{m_i}{q_i}\sum_{i\in \mathcal{K} ^{(t-1)}(\boldsymbol{q})}{\nabla _{\mathbf{w}}}f_i(\mathbf{w}_{i}^{t-1})||^2 ] 
\notag\\=&
\mathbb{E} [ ||\nabla _{\mathbf{h}_s}f(\mathbf{h}_{s}^{t-1})-\frac{1}{K}\frac{m_i}{q_i}\sum_{i\in \mathcal{K} ^{(t-1)}(\boldsymbol{q})}{\nabla _{\mathbf{h}_s}}f_i(\mathbf{h}_{s,i}^{t-1})||^2 ] +
\notag\\
&\mathbb{E} [ ||\nabla _{\mathbf{h}_c}f(\mathbf{h}_{s}^{t-1})-\frac{1}{K}\frac{m_i}{q_i}\sum_{i\in \mathcal{K} ^{(t-1)}(\boldsymbol{q})}{\nabla _{\mathbf{h}_s}}f_i(\mathbf{h}_{s,i}^{t-1})||^2 ],
\end{align}
and (b) follows from {\bf Assumption 2} together with the  inequality~\eqref{difference_wc} and~\eqref{difference_ws}, namely
{\begin{align}\label{difference_wc}
&\mathbb{E}[ \Vert \nabla_{{{\bf{h}}_c}} f({\mathbf{h}}_c^{t-1}) -
\frac{1}{K}\frac{m_i}{q_i}\sum_{i\in \mathcal{K} ^{(t-1)}(\boldsymbol{q})}{\nabla _{\mathbf{h}_c}}f_i(\mathbf{h}_{c,i}^{t-1})
\Vert^{2}] \nonumber \\
=&\mathbb{E} [||\frac{1}{K}\frac{m_i}{q_i}\sum_{i\in \mathcal{K} ^{(t-1)}(\boldsymbol{q})}{\nabla _{\mathbf{h}_c}}f_i(\mathbf{h}_{c}^{t-1})-\nonumber \\
& \quad \quad \quad \quad \quad \quad \quad \quad \quad \quad \frac{1}{K}\frac{m_i}{q_i}\sum_{i\in \mathcal{K} ^{(t-1)}(\boldsymbol{q})}{}\nabla _{\mathbf{h}_c}f_i(\mathbf{h}_{c,i}^{t-1})||^2]\nonumber \\
\leq& \frac{1}{K} \mathbb{E} [ 
\sum_{i\in \mathcal{K} ^{(t-1)}(\boldsymbol{q})}{\frac{{m_i}^2}{{q_i}^2}||}\nabla _{\mathbf{h}_c}f_i(\mathbf{h}_{c}^{t-1})-\nabla _{\mathbf{h}_c}f_i(\mathbf{h}_{c,i}^{t-1})||^2
] \nonumber \\
\leq& 
\beta ^2\frac{K}{K}\mathbb{E} _{\mathcal{K} ^{(t-1)}(\boldsymbol{q})}[\frac{{m_i}^2}{{q_i}^2}||\mathbf{h}_{c}^{t-1}-\mathbf{h}_{c,i}^{t-1}||^2] 
 \nonumber \\
\leq & 
\sum_{i=1}^N{}\frac{\beta ^2{ m_i}^22\gamma ^2I^2}{q_i}\big(  N\underset{i^\prime}{\max}\big\{ \frac{{m_{i^\prime}}^2}{q_{i^\prime}}\frac{1}{\left( 1-a_{i^\prime} \right)}\frac{1}{\left( 1-p_{i^\prime} \right) \left( 1-\varphi_{i^\prime} \right)} \big\} \notag\\&+\frac{1}{\left( 1-a_i \right)}\frac{1}{\left( 1-p_i \right) \left( 1-\varphi _i \right)} \big) \sum_{j=1}^{L_c}{G_{j}^{2}} 
,
\end{align}
}%
where the first inequality uses $\Vert \sum_{i=1}^{N} \mathbf{z}_{i}\Vert^{2} \leq N \sum_{i=1}^{N} \Vert \mathbf{z}_{i}\Vert^{2}$ for arbitrary vectors $\mathbf{z}_{i}$; the second inequality leverages the $\beta$-smooth of each $f_{i}$ by {\bf Assumption 1}; and the third inequality is derived according to {\bf Lemma 1}, and
 \begin{align}\label{difference_ws}
&\mathbb{E}[ \Vert \nabla_{{{\bf{h}}_s}} f({\mathbf{h}}_s^{t-1}) - \frac{1}{K}\sum\limits_{i \in \mathcal{K}^{(t -1)}(\bm q)} \nabla_{{{\bf{h}}_s}} f_{i} (\mathbf{h}_{s,i}^{t-1})\Vert^{2}] \nonumber \\
\leq& 
\sum_{i=1}^N{}\beta ^2\frac{{m_i}^2}{q_i}\mathbb{E} [||\mathbf{h}_{s}^{t-1}-\mathbf{h}_{s,i}^{t-1}||^2]
 \overset{(a)}{=} 0,
\end{align}
where (a) holds since the server-side sub-models are aggregated every round. Thus, at any round $t$, each client’s server-side sub-model coincides with the aggregated server-side model.

Dividing Eq.~\eqref{eq:11} both sides by $\frac{\gamma}{2}$ and rearranging the  terms yields
{ \begin{align}
&\mathbb{E}\left [\Vert \nabla_{\bf{w}} f({\mathbf{w}}^{t-1})\Vert^{2}\right] \nonumber \\
\leq& \frac{2}{\gamma} \left(\mathbb{E}\left[f({\mathbf{w}}^{t-1})\right] \! - \!\mathbb{E}\left[f({\mathbf{w}}^{t})\right]\right)+ \sum_{i=1}^N\frac{{m_i}^2}{q_i}\cdot
 \notag \\
&
\frac{\beta\gamma}{\left( 1-p_i \right) }\bigg\{\frac{1}{\left( 1-\varphi _i \right)}\frac{1}{\left( 1-a_i \right)}\sum_{j=1}^{L_c}({\sigma _{j}^{2}}+{G_j^2})+\sum_{j=L_c+1}^L({\sigma _{j}^{2}}+
\notag \\
&{G_j^2})\bigg\}+\sum_{i=1}^N{}\beta ^2\frac{{m_i}^2}{q_i}2\gamma ^2I^2 \big( N\underset{i^\prime}{\max}\big\{ \frac{{m_{i^\prime}}^2}{q_{i^\prime}}\frac{1}{\left( 1-p_{i^\prime} \right) }\frac{1}{ \left( 1-\varphi_{i^\prime} \right)}  \notag\\
&\frac{1}{\left( 1-a_{i^\prime} \right)}\big\}+\frac{1}{\left( 1-a_i \right)}\frac{1}{\left( 1-p_i \right) \left( 1-\varphi _i \right)} \big) \sum_{j=1}^{L_c}{G_{j}^{2}}.
\label{eq:pf-thm-rate-eq8}
\end{align}
}%

Summing over $t\in\{1,\ldots, R\}$ and dividing by $R$ gives 
\begin{align}
&\frac{1}{R} \sum_{t=1}^{R} \mathbb{E}\left [\Vert \nabla_{\bf{w}} f({\mathbf{w}}^{t-1})\Vert^{2}\right] \notag
\\\overset{(a)}{\leq} 
&  \frac{2}{\gamma R} \left(f({\mathbf{w}}^{0}) -f^{\ast}\right)  + \sum_{i=1}^N\frac{{m_i}^2}{q_i}\frac{\beta\gamma}{\left( 1-p_i \right) }\cdot
 \notag \\
&
\bigg\{\frac{1}{\left( 1-\varphi _i \right)}\frac{1}{\left( 1-a_i \right)}\sum_{j=1}^{L_c}({\sigma _{j}^{2}}+{G_j^2})+\sum_{j=L_c+1}^L({\sigma _{j}^{2}}+
{G_j^2})\bigg\}
+
\notag \\
&\sum_{i=1}^N{}\beta ^2\frac{{m_i}^2}{q_i}2\gamma ^2I^2 \big( N\underset{i^\prime}{\max}\big\{ \frac{1}{\left( 1-p_{i^\prime} \right) \left( 1-\varphi_{i^\prime} \right) \left( 1-a_{i^\prime} \right) }\cdot \notag\\
& \frac{{m_{i^\prime}}^2}{q_{i^\prime}}\big\}+\frac{1}{\left( 1-a_i \right)}\frac{1}{\left( 1-p_i \right) \left( 1-\varphi _i \right)} \big) \sum_{j=1}^{L_c}{G_{j}^{2}} \notag\\
\overset{(b)}{\leq} 
&  \frac{2}{\gamma R} \left(f({\mathbf{w}}^{0}) -f^{\ast}\right) 
+ \sum_{i=1}^N\frac{{m_i}^2}{q_i}\frac{\beta\gamma}{\left( 1-p_i \right) } \cdot
 \notag \\
&
\bigg\{\frac{1}{\left( 1-\varphi _i \right)}\frac{1}{\left( 1-a_i \right)}\sum_{j=1}^{L_c}({\sigma _{j}^{2}}+{G_j^2})+\sum_{j=L_c+1}^L({\sigma _{j}^{2}}+
{G_j^2})\bigg\}
+
\notag \\
&\sum_{i=1}^N{}\beta ^2\frac{{m_i}^2}{q_i}2\gamma ^2I^2 \big( N\underset{i^\prime}{\max}\big\{ \frac{1}{\left( 1-p_{i^\prime} \right) \left( 1-\varphi_{i^\prime} \right) \left( 1-a_{i^\prime} \right) }\cdot \notag\\
& \frac{{m_{i^\prime}}^2}{q_{i^\prime}}\big\}+\frac{1}{\left( 1-a_i \right)}\frac{1}{\left( 1-p_i \right) \left( 1-\varphi _i \right)} \big) \sum_{j=1}^{L_c}{G_{j}^{2}},
\end{align}
where (a) follows from $f^{\ast}$ is the minimum value of  Problem \textbf{P1}.

\bibliographystyle{IEEEtran}
\bibliography{citationlist}
\vspace{-0.4cm}
\begin{IEEEbiographynophoto}{Wei Wei}(Graduate Student Member, IEEE)
received the B.S. degree in electronics engineering from Shandong University, Jinan, China, in 2020, and the M.S. degree (Hons.) in electronic and information engineering from Tsinghua University, Beijing, China, in 2023. She is currently pursuing the Ph.D. degree with the Department of Electrical and Electronic Engineering, the University of Hong Kong, Hong Kong, China. Her research interests include edge computing and multiagent reinforcement learning.
\end{IEEEbiographynophoto}
\vspace{-0.4cm}
\begin{IEEEbiographynophoto}{Zheng Lin}(Graduate Student Member, IEEE)
received the B.Eng. degree in electronic information from Fuzhou University in 2020, and the M.Eng. degree in electrical and computer engineering from Fudan University in 2023. He is currently pursuing his Ph.D. degree with the Department of Electrical and Electronic Engineering, the University of Hong Kong, Hong Kong, China. He was awarded Fudan Outstanding Graduate and Shanghai Outstanding Graduate in 2023. He serves as a TPC member of several international conferences and a Young Professional in IEEE Vehicular Technology Society Ad Hoc Committee. His research interests include wireless networking, edge intelligence, and distributed machine learning.
\end{IEEEbiographynophoto}

\vspace{-0.4cm}
\begin{IEEEbiographynophoto}{Xihui Liu}(Member, IEEE)
received the bachelor’s degree from Tsinghua University and the PhD degree from the Chinese University of Hong Kong. She is an assistant professor with the Department of Electrical and Electronic Engineering and the Institute of Data Science, the University of Hong Kong. Before joining HKU, she was a postdoctoral researcher with the University of California, Berkeley. Her research interests include computer vision, deep learning, generative models, and multimodal AI. She was awarded the Adobe Research Fellowship in 2020, named an EECS Rising Star in 2021, and received the WAIC Rising Star Award in 2022. She serves as an area chair for CVPR 2024, ACM MM 2024, and ICLR 2025.
\end{IEEEbiographynophoto}

\vspace{-0.4cm}
\begin{IEEEbiographynophoto}{Hongyang Du}
is an assistant professor with the Department of Electrical and Computer Engineering, The University of Hong Kong, where he directs the Network Intelligence and Computing Ecosystem (NICE) Laboratory. He received the B.Eng. degree from Beijing Jiaotong University, China, and the Ph.D. degree from Nanyang Technological University, Singapore. He serves as an Editor of IEEE Communications Surveys \& Tutorials, IEEE Transactions on Communications, IEEE Transactions on Vehicular Technology, and IEEE Open Journal of the Communications Society, and as a Guest Editor of IEEE Vehicular Technology Magazine. He is the recipient of the IEEE ComSoc Young Professional Award for Best Early Career Researcher in 2024, the IEEE Daniel E. Noble Fellowship Award from the IEEE Vehicular Technology Society in 2022, the IEEE Signal Processing Society Scholarship in 2023, and the Singapore Data Science Consortium Dissertation Research Fellowship in 2023. He was recognized as an exemplary reviewer of the IEEE Transactions on Communications and IEEE Communications Letters. His research interests include edge intelligence, generative AI, and network management.
\end{IEEEbiographynophoto}

\vspace{-0.4cm}
\begin{IEEEbiographynophoto}{Dusit Niyato}(Fellow, IEEE)
received the B.Eng. degree from King Mongkut’s Institute of Technology Ladkrabang (KMITL), Thailand, in 1999, and the Ph.D. degree in electrical and computer engineering from the University of Manitoba, Canada, in 2008. He is currently a Professor with the School of Computer Science and Engineering, Nanyang Technological University, Singapore. His research interests include the Internet of Things (IoT), machine learning, and incentive mechanism design.
\end{IEEEbiographynophoto}

\vspace{-0.4cm}
\begin{IEEEbiographynophoto}{Xianhao Chen}(Member, IEEE)
received the B.Eng. degree in electronic information from Southwest Jiaotong University in 2017, and the Ph.D. degree in electrical and computer engineering from the University of Florida in 2022. He is currently an assistant professor with the Department of Electrical and Electronic Engineering, the University of Hong Kong, where he directs the Wireless Information \& Intelligence (WILL) Lab. He serves as a TPC member of several international conferences and on the editorial board of ACM Computing Surveys, IEEE Transactions on Networking, and IEEE Transactions on Vehicular Technology. He received the Early Career Award from the Research Grants Council (RGC) of Hong Kong in 2024, the ECE Graduate Excellence Award for Research from the University of Florida in 2022, and the ICCC Best Paper Award in 2023. His research interests include wireless networking, edge intelligence, and machine learning.
\end{IEEEbiographynophoto}
\end{document}